\documentclass[runningheads]{llncs}

\newif\ifextendedversion
\extendedversiontrue

\usepackage[T1]{fontenc}
\usepackage{hyperref}
\usepackage[obeyFinal]{todonotes}
\usepackage{amsmath}
\usepackage{amssymb}
\usepackage{booktabs,colortbl}
\usepackage[numbers,sort&compress]{natbib}
\usepackage[capitalize]{cleveref}
\usepackage{siunitx}
\usepackage{mathtools}
\usepackage[new]{old-arrows}
\usepackage{stmaryrd}
\usepackage{fontawesome}
\usepackage[newfloat,frozencache=true,cachedir=minted-cache]{minted}
\usepackage{halloweenmath}
\usepackage{centernot}
\usepackage[caption=false]{subfig}
\usepackage{cleveref}
\usepackage{chngcntr}
\usepackage{xspace}
\usepackage{siunitx}
\usepackage{enumitem}
\usepackage[nounderscore]{syntax}
\usepackage{multicol}
\usepackage{cancel}
\usepackage{lineno}
\usepackage[misc,geometry]{ifsym}

\ifextendedversion
	\usepackage[bibliography=common]{apxproof}

\else
	\usepackage[appendix=strip]{apxproof}
\fi

\usepackage{float}
\usepackage{wrapfig}

\newcommand{\Threads}{\ensuremath{T}}
\newcommand{\newtid}{\ensuremath{\nu}}
\newcommand{\Threadprototypes}{\ensuremath{T_{\mathit{templ}}}}
\newcommand{\Nodes}{\ensuremath{\mathcal{N}}}
\newcommand{\st}{\textsf{st}}
\newcommand{\InterleavingExecutions}{\ensuremath{\mathcal{I}}}
\newcommand{\LocalTraces}{\ensuremath{\mathcal{LT}}}
\newcommand{\GlobalTraces}{\ensuremath{\mathcal{GT}}}
\newcommand{\Edges}{\ensuremath{\mathcal{E}}}
\newcommand{\States}{\ensuremath{\mathcal{S}}}

\newcommand{\Globals}{\ensuremath{\mathcal{G}}}
\newcommand{\Values}{\ensuremath{\mathcal{V}}}
\newcommand{\Mutexes}{\ensuremath{\mathcal{M}}}
\newcommand{\MutexesForGlobals}{\Mutexes_{\Globals}}
\newcommand{\LocalStates}{\ensuremath{\mathcal{L}}}
\newcommand{\LocalVariables}{\ensuremath{\mathcal{X}}}
\newcommand{\GhostLocalVariables}{\ensuremath{\LocalVariables}_{\ghost}}

\newcommand{\loctrace}{\ensuremath{\mathit{lt}}}
\newcommand{\globtrace}{\ensuremath{\mathit{gt}}}
\newcommand{\interleaving}{\ensuremath{\mathit{i}}}
\newcommand{\multithreaded}{\ensuremath{\mathit{multithreaded}}}

\newcommand{\atomic}{\code{atomic}}
\newcommand{\create}[1]{\code{create(#1)}}
\newcommand{\lock}[1]{\code{lock(#1)}}
\newcommand{\unlock}[1]{\code{unlock(#1)}}
\newcommand{\assert}[1]{\code{assert(#1)}}
\newcommand{\Language}{\ensuremath{\textsf{Lang}}}
\newcommand{\LanguageWithAtomic}{\ensuremath{\textsf{Lang}_{\mathit{Atomic}}}}
\newcommand{\LanguageWithMg}{{\textsf{Lang}_{\MutexesForGlobals}}}
\newcommand{\ghost}{\mathghost}
\newcommand{\deghostifyState}{\ensuremath{\pi_{P}}}
\newcommand{\deghostifyInterleaving}{\ensuremath{\rho}}
\newcommand{\ghostifyUpdate}{\ensuremath{\delta_W}}
\newcommand{\ghostifyInvariant}{\ensuremath{\Psi_W}}
\newcommand{\GlobalDeclarations}{\ensuremath{\mathcal{D}}}
\newcommand{\GhostGlobalDeclarations}{\ensuremath{\mathcal{D}}_{\ghost}}
\newcommand{\GhostUpdate}{\ensuremath{\mathbb{U}}}
\newcommand{\GhostInvariant}{\ensuremath{\mathbb{I}}}
\newcommand{\Witness}{\ensuremath{(\GhostGlobalDeclarations, \GhostLocalVariables, \GhostUpdate, \GhostInvariant)}}

\newcommand{\semI}[1]{\llbracket#1\rrbracket_{\InterleavingExecutions}}
\newcommand{\semLT}[1]{\llbracket#1\rrbracket_{\LocalTraces}}
\newcommand{\semGT}[1]{\llbracket#1\rrbracket_{\GlobalTraces}}
\newcommand{\threadsI}{\textsf{threads}_\InterleavingExecutions}
\newcommand{\threadsGT}{\textsf{threads}_\GlobalTraces}

\newcommand{\code}[1]{\texttt{#1}}
\newcommand{\tool}[1]{\textsc{#1}}
\newcommand{\main}{\code{main}}

\newcommand{\Ultimate}{\tool{Ultimate}\xspace}
\newcommand{\UltimateGemCutter}{\tool{Ultimate GemCutter}\xspace}
\newcommand{\GemCutter}{\tool{GemCutter}\xspace}
\newcommand{\Goblint}{\tool{Goblint}\xspace}

\SetupFloatingEnvironment{listing}{name=List.,within=none}
\crefname{listing}{Listing}{Listings}
\Crefname{listing}{Listing}{Listings}
\makeatletter
\@ifundefined{c@sublisting}{\newsubfloat[captionskip=10pt]{listing}}{}
\makeatother
\makeatletter
  \newbox\sf@box
  \newenvironment{SubFloat}[2][]%
    {\def\sf@one{#1}%
     \def\sf@two{#2}%
     \setbox\sf@box\hbox
       \bgroup}%
    {  \egroup
     \ifx\@empty\sf@two\@empty\relax
       \def\sf@two{\@empty}
     \fi
     \ifx\@empty\sf@one\@empty\relax
       \subfloat[\sf@two]{\box\sf@box}%
     \else
       \subfloat[\sf@one][\sf@two]{\box\sf@box}%
     \fi}
\makeatother
\definecolor{codehighlight}{RGB}{230,250,250}
\setminted{frame=single, framerule=0.5pt, framesep=5pt, numbersep=0.5em, xleftmargin=1.25em,
           tabsize=2, mathescape=true, escapeinside=||, highlightcolor=codehighlight, numberblanklines=false,
           baselinestretch=1, fontsize=\footnotesize}
\definecolor{operator}{RGB}{164, 0, 0}

\newcommand{\codecreate}[1]{\textcolor{operator}{create}(${#1}$)}
\newcommand{\codelock}[1]{\textcolor{operator}{lock}({#1})}
\newcommand{\codeunlock}[1]{\textcolor{operator}{unlock}({#1})}
\newcommand{\codeassert}[1]{\textcolor{operator}{assert}({#1})}
\definecolor{atomic}{RGB}{100, 102, 98}
\newcommand{\codeatomicbegin}{\textcolor{atomic}{atomic}\{}
\newcommand{\codeatomicend}{\}}
\newcommand{\codehighlight}[1]{\colorbox{codehighlight}{#1}}

\definecolor{colormain}{HTML}{204a87}
\definecolor{colort1}{RGB}{164, 0, 0}
\newcommand{\ltDrawingDefs}{
    \usetikzlibrary{calc}
    \usetikzlibrary{fit, shapes.geometric, arrows, positioning}
    \tikzset{
        every node/.style={node distance=50pt and 45pt},
        programpoint/.style={circle,draw,font=\small},
        programpointsink/.style={programpoint,line width=0.5mm},
        programpointwide/.style={programpoint,node distance=40pt and 55pt},
        programpointwidesink/.style={programpointwide,line width=0.5mm},
        edgelabel/.style={midway,font=\small,above=0.5mm}
    }
    \newcommand{\programo}[3]{
        \draw[->](##1)--(##3)node[edgelabel]{$##2$};
    }
    \newcommand{\programon}[4][]{
        \node[programpoint##1,right= of ##2](##4){};
        \draw[->](##2)--(##4)node[edgelabel]{$##3$};
    }
    \newcommand{\programondo}[4][]{
      \coordinate[right= of ##2, xshift=-20pt](##4i){};
      \draw[-,dotted](##2)--(##4i)node[]{$##3$};
      \node[programpoint##1,right= of ##4i, xshift=-30pt](##4){};
      \draw[->](##4i)--(##4)node[edgelabel]{$##3$};
    }
    \newcommand{\programond}[4][]{
        \node[programpoint##1,right= of ##2,dashed](##4){};
        \draw[->,dashed](##2)--(##4)node[edgelabel]{$##3$};
    }
    \newcommand{\programont}[5][]{
        \node[programpoint##1,right= of ##2, label distance=0mm, label={-176:${##4}$}](##5){};
        \draw[->](##2)--(##5)node[edgelabel]{$##3$};
    }
    \newcommand{\createo}[3][]{
        \draw[-latex,blue](##2)--(##3)node[edgelabel,##1]{$\to_c$};
    }
    \newcommand{\mutexo}[5][]{
        \draw[-latex,red,##5](##2) to node[edgelabel,##1]{$\to_{##4}$} (##3);
    }
}

\newcommand{\semarrow}[2]{
\lhook\joinrel\xrightarrow{#2}_{#1}
}
\newcommand{\semarrowLT}[2]{
\joinrel\xrightarrow{#2}_{#1}
}

\newcommand{\splitAtomic}{\textsf{split}}
\newcommand{\extend}{\textsf{extend}}
\newcommand{\projectSplitAtomic}{\ensuremath{\pi_{\splitAtomic}}}
\newcommand{\projectGhostify}{\ensuremath{\pi_{\ghost}}}

\newcommand{\partialto}{\rightharpoonup}
\newcommand{\undefined}{\uparrow}
\newcommand{\naturalnumbers}{\mathbb{N}}
\newcommand{\self}{\texttt{self}}

\newcommand{\ttbox}[1]{\colorbox{black!8}{\footnotesize\texttt{#1}\tiny\strut}}
\newlength\multilen

\makeatletter
  \def\myrulefill{\leavevmode\leaders\hrule height .7ex width 1ex depth -0.6ex\hfill\kern\z@}
\makeatother
\newcommand{\separator}[1]{\multicolumn{3}{c}{\hrulefill\quad #1 \quad\hrulefill\hrulefill}}

\usepackage{xintfrac}
\makeatletter
\newcommand\Duration [1]{
  \xintifLt{#1}{60}
  {
      \xintRound {1}{#1}
  }
  {
    \expandafter\Duration@a
    \romannumeral0\xintiidivision{\xintiRound{0}{#1}}{60}%
  }%
}%

\newcommand\Duration@a [2]{
   \expandafter\Duration@b
   \romannumeral0\xintiidivision{#1}{60}:\Duration@format{#2}%
}%

\newcommand\Duration@b [2]{
   \xintiiifZero {#1}
      {#2}
      {#1:\Duration@format {#2}}
}%

\newcommand\Duration@format [1]{\expandafter\@gobble\the\numexpr 100+#1\relax}

\makeatother

\newcommand{\hurl}[1]{{\scriptsize\UrlFont\href{https://#1}{\path{#1}}}}

\makeatletter
\newcommand{\mylabel}[1]{%
    \@ifundefined{c@#1}{%
        \newcounter{#1}%
        \setcounter{#1}{0}%
    }{}%
    \ifthenelse{\value{#1} > 0}{}{%
        \label{#1}%
        \addtocounter{#1}{1}%
    }%
}
\makeatother

\ifextendedversion
\else
\usepackage[firstpage]{draftwatermark}
\SetWatermarkAngle{0}
\SetWatermarkText{\hspace*{5.5in}\raisebox{8.8in}{\href{https://doi.org/10.5281/zenodo.13863579}{\includegraphics{badge-afr}}}}
\fi
\begin{document}
\title{Correctness Witnesses for Concurrent Programs:\\
Bridging the Semantic Divide with Ghosts \ifextendedversion(Extended Version)\fi}
\titlerunning{Correctness Witnesses for Concurrent Programs}
%
%
\author{Julian~Erhard \inst{1,5}\textsuperscript{(\Letter)}
	\and Manuel~Bentele \inst{2,6}
	\and Matthias~Heizmann \inst{4}
	\and Dominik~Klumpp \inst{2}
	\and Simmo~Saan \inst{3}
	\and Frank~Schüssele \inst{2}
	\and Michael~Schwarz \inst{1}
	\and Helmut~Seidl \inst{1}
	\and Sarah~Tilscher \inst{1,5}
	\and Vesal~Vojdani \inst{3}
}

\institute{
    Technical University of Munich, Garching, Germany \and
    University of Freiburg, Freiburg, Germany \\
	\and
    University of Tartu, Tartu, Estonia \and
    University of Stuttgart, Stuttgart, Germany \and
    Ludwig-Maximilians-Universität München, Munich, Germany \and
    Hahn-Schickard, Villingen-Schwenningen, Germany \\
	\email{julian.erhard@tum.de}
}

\authorrunning{J. Erhard et al.}
%
%
\maketitle
\begin{abstract}
	Static analyzers are typically complex tools and thus prone to contain bugs themselves.
	To increase the trust in the verdict of such tools, \emph{witnesses} encode key reasoning steps underlying the verdict in an exchangeable format, enabling independent validation of the reasoning by other tools.
	For the correctness of concurrent programs, no agreed-upon witness format exists --- in no small part due to the divide between the semantics considered by analyzers, ranging from interleaving to thread-modular approaches,
	making it challenging to exchange information.
	We propose a format that leverages the well-known notion of \emph{ghosts} to embed the claims a tool makes about a program into a modified program with ghosts, such that the validity of a witness
	can be decided by analyzing this program.
	Thus, the validity of witnesses with respect to the interleaving and the thread-modular semantics coincides.
	Further, thread-modular invariants computed by an abstract interpreter can naturally be expressed in the new format using ghost statements.
	We evaluate the approach by generating such ghost witnesses for
	a subset of concurrent programs from the SV-COMP benchmark suite,
	and pass them to a model checker.
	It can confirm \qty{75}{\percent} of these witnesses ---
	indicating that ghost witnesses can bridge the semantic divide between interleaving and thread-modular approaches.
\keywords{Software Verification \and Correctness Witnesses \and Concurrency \and Ghost Variables \and Abstract Interpretation \and Model Checking.}
\end{abstract}

\section{Introduction}

While static analysis tools can help developers write bug-free programs, these tools sometimes return incorrect verdicts due to bugs in the tools themselves.
To increase the trust in the verdicts of static analysis tools,
\emph{witnesses} have been proposed as artifacts that contain further information about static analysis results~\cite{Beyer2015,Beyer2016}.
For indicating that a program property does not hold, a witness may exhibit a program execution that violates the property (\emph{violation witness}).
For substantiating that a property holds throughout all possible executions, a suitable set of program invariants may be provided
that guides the static analysis tool towards proving the property of interest (\emph{correctness witness}) --- and may expose errors in reasoning when invariants can be shown to be violated.
\emph{Validators} are static analysis tools that consume witnesses and report whether they can re-establish verdicts.
They are, e.g., used in SV-COMP \cite{Beyer2023}.
Here, we are interested in a witness format suitable for certifying the correctness of concurrent programs.
We extend the format for correctness witnesses by \citet{Ayaziova2024}, which allows expressing invariants per location, with \emph{ghost variables}.
Ghost variables have not only been proposed for the verification of sequential programs \cite{Schreiber97} but have also been employed for concurrent programs \cite{OwickiG76,Keller76,AptBO09}.
Ghost variables --- sometimes referred to as auxiliary variables --- are additional program variables introduced to ease the specification and verification of intricate program properties.
These variables allow encoding the progress of other threads, making it possible to relate observations that the current thread may make to this progress.
A witness then specifies a set of ghost variables and how they evolve via \emph{ghost updates} inserted at existing program locations.
Additionally, it contains \emph{invariants}, which can refer both to program and ghost variables and, thus, to properties that may be difficult to express without ghosts.
\begin{listing}[h]
  \centering
  \begin{SubFloat}{\label{listing:running-example-original}Concurrent program.}
    \begin{minipage}[t]{0.3\linewidth}
      \begin{minted}[frame=none, linenos]{c}
unsigned int used = 0;

|$main$|:
  |\codecreate{t_1}|;

  |\codelock{m}|;
    |\codeassert{used == 0}|;
  |\codeunlock{m}|;

|$t_1$|:
  |\codelock{m}|;
    used = 47;
    used = 0;
  |\codeunlock{m}|;
      \end{minted}
    \end{minipage}
  \end{SubFloat}
  \hfill
  \begin{SubFloat}{\label{listing:running-example-annotated}Concurrent program with ghosts.}
    \begin{minipage}[t]{0.61\linewidth}
      \setlength\fboxsep{1.25pt}
      \let\oldFancyVerbLine\theFancyVerbLine
      \renewcommand{\theFancyVerbLine}{
        \ifnum\value{FancyVerbLine}=0
        \else
          \oldFancyVerbLine
        \fi
      }
      \begin{minted}[frame=none, linenos, breaklines]{c}
unsigned int used = 0;

|$main$|:
  |\codecreate{t_1}|;
  |\codehighlight{\codeatomicbegin \codeassert{$\ghost$ == 0 $\Longrightarrow$ used == 0}; \codeatomicend}\setcounter{FancyVerbLine}{0}|
  |\codehighlight{\codeatomicbegin}| |\codelock{m}|; |\codehighlight{$\ghost$ = 1; \codeatomicend}\setcounter{FancyVerbLine}{4}|
    |\codeassert{used == 0}|;
  |\codehighlight{\codeatomicbegin}| |\codeunlock{m}|; |\codehighlight{$\ghost$ = 0; \codeatomicend}|

|$t_1$|:
  |\codehighlight{\codeatomicbegin}| |\codelock{m}|; |\codehighlight{$\ghost$ = 1; \codeatomicend}|
    used = 47;
    used = 0;
  |\codehighlight{\codeatomicbegin}| |\codeunlock{m}|;| \codehighlight{$\ghost$ = 0; \codeatomicend}|
      \end{minted}
    \end{minipage}
  \end{SubFloat}
  \caption{Example program without \protect\subref{listing:running-example-original} and with ghost statements \protect\subref{listing:running-example-annotated}.
  Parts highlighted in blue show ghost statements from a witness.
  }\label{listing:running-example}
\end{listing}

\begin{example}\label{e:running-example}
	The program from Listing \ref{listing:running-example-original} is an example of a \emph{resource invariant}:
	the variable $\mathtt{used}$ is always $0$ except when thread $t_1$ holds the mutex \code{m}.
	We cannot state this property with C assertions.
	Instead, in Listing~\ref{listing:running-example-annotated} we add a ghost variable $\ghost$ and maintain it such that it indicates whether the mutex \code{m} is currently held by any thread.
	The assertion in $main$ states that, when no thread is in a critical section, the value of the global variable $\mathtt{used}$ is $0$.
\end{example}
A core requirement for a witness format is that it facilitates the exchange of information between tools.
Different tools employ different semantics to formalize the behavior of (sequentially consistent) concurrent programs.
Some may use an interleaving semantics \cite{Lamport79}, while others turn to a thread-modular semantics \cite{Lamport78,Mukherjee2017,Schwarz2021}.
\ After showing that safety of programs coincides for the interleaving and a thread-modular semantics (\cref{sec:programs-and-semantics}),
we introduce a witness format that allows instrumenting a program with ghost statements (\cref{sec:ghost-witnesses}) that are executed atomically with existing statements.
We show how, by this construction, the validity of witnesses with respect to the interleaving and the thread-modular semantics also coincides (\cref{sec:valid-witness-local-traces}).
We exemplify how ghost witnesses are naturally suited to express information obtained for concurrent programs by a thread-modular\ abstract interpreter (\cref{sec:generation}).
Our format for ghost witness (\cref{sec:format}) extends the existing \textsc{SV-COMP} witness format, easing adoption by other software verifiers.
For the experimental evaluation, we automatically generate witnesses using \Goblint~\cite{Vojdani2016,Saan2023}, and validate them using the model checker \UltimateGemCutter~\cite{GemCutterSVCOMP,GemCutterPLDI} (\cref{sec:validation-with-model-checking}).

\section{Two Views on Concurrent Programs}\label{sec:programs-and-semantics}
We first introduce the notion of programs that we will consider and present an interleaving and a thread-modular semantics for those programs.
To simplify the presentation, we model a core subset of the C~language, which is later extended by allowing the insertion of ghost statements.
The language supports dynamic thread creation, locking and unlocking of mutexes, and reading and writing of global variables.
Function calls, pointers and heap memory are skipped to not overcomplicate the exposition.

\subsection{Programs}\label{subsec:programs}
A program is given by a set of global variables $\Globals$, a set of mutexes $\Mutexes$, a set of local states $\LocalStates$ and a finite number of named control-flow graphs $\Threadprototypes$, called \emph{thread templates}, one of which is \main.
We demand that $\Globals \cap \Mutexes = \emptyset$, and further, that there is a set of global declarations $\GlobalDeclarations$ that provides types and initial values for all global variables.
Local states are a type-correct mapping from local variables $\LocalVariables$ to values $\Values$.
For simplicity, we assume that all threads use the same set of local variables, and we omit procedures.
Each control-flow graph consists of a finite set of nodes $\Nodes$ and labeled edges \Edges.
The sets of program points are disjoint between the different control-flow graphs.
An edge $(u, a, v) \in \Edges$ consists of a source node $u$, an action $a$, and a sink node $v$.
We demand that for a given $u, v$, there may be at most one $a$, such that $(u, a, v) \in \Edges$.
Each control-flow graph has a dedicated initial node with no incoming edges, at which execution is meant to start.
The following statements are supported:

\begin{description}
	\item[Locking/Unlocking mutexes.] The actions \code{lock(m)} and \code{unlock(m)}, where \code{m} is some mutex.
	\item[Thread creation.] The action \code{create(t)}, where \code{t} is the name of some thread template.
	\item[Local Update.] The action \code{l := f l}, where \code{f} is a pure function taking the local state $l$, yielding the new local state of the thread.
	\item[Global Read.] The action \code{l := f l g},  where \code{f} is a  pure function taking the local state $l$ and the value of some global variable \code{g}, yielding the new local state of the thread.
	\item[Global Write.] The action \code{g := f l}, where \code{f} is some pure function taking the local state $l$, updating the global variable.
	\item[Assertion.] The action \code{assert(p l)}, where \code{p} is some pure function mapping the local state $l$ to a boolean value.
	\item[Guard.] The actions \code{Pos(c l)} and \code{Neg(c l)}, realizing branching on some condition $c$ over the local state $l$, where $c$ is a pure function yielding a boolean.
\end{description}
We refer to the language of programs with these actions as $\Language$.

\subsection{Interleaving Semantics}
\label{sec:interleaving-semantics}
The interleaving semantics is a concrete semantics of concurrent programs where all actions in an execution are totally ordered.
To distinguish threads,
they are identified with a thread id computed from their creation history.
The set of all possible thread ids for a program, in the following often simply called threads, is referred to as $\Threads$.
The thread id of the initial thread is the empty sequence, while for a created thread, the thread id is obtained by concatenating the thread id of the parent thread with the number of threads the parent thread has created before.
A program configuration (or state) $(L, M, G) \in \States: (\Threads \partialto (\Nodes \times \LocalStates)) \times (\Mutexes \partialto \Threads) \times (\Globals \rightarrow \Values)$ is a triple, where $L$ is a partial mapping of threads $\Threads$ to program nodes $\Nodes$ and local states $\LocalStates$, $M$ is a partial mapping from mutexes to the threads that hold them,
and $G$ is a mapping of global variables to values.
Each thread has a local variable $\self$ that contains its thread id that is only set when the thread is created.
We write $M\,m\undefined$, if the value of $M$ is not defined for $m$, i.e., the mutex $m$ is not held by any thread in $M$.
An execution of a program is given by a sequence of program configurations $(s_i)_{1\leq i \leq k}$, with $s_i \in \States$, for some natural number $k \geq 1$, interleaved with a sequence $((e_i, t_i))_{1 \leq i \leq k -1}$, with $e_i \in \Edges, t_i \in \Threads$, of edges annotated with the thread taking them.
For notational convenience, let us denote an execution step via an edge $(u, a, v)$ from state $s_i$ to state $s_{i+1}$ taken be thread $t$, as follows:
\[
 s_i \semarrow{t}{(u, a, v)} s_{i+1}
\]
We demand that the interleavings are \emph{consistent}.
The set of consistent interleavings is defined inductively:
\begin{enumerate}[label={I\arabic*}, ref={I\arabic*}]
  \item \label{interleaing:initial} The sequence consisting of some initial state
  $
  s_0 = (\{t_0 \mapsto (\st_{\main}, l_0, 0)\},\emptyset,g_0)
  $
  where $l_0$ is the initial local state, i.e., maps $\self$ to the initial thread id $t_0$ and all other local variables to some initial value $0$, and $g_0$ maps all globals to their initial value according to the global declarations $\GlobalDeclarations$, is consistent.
 \item  \label{interleaving:extending}
 Let $i$ be a consistent interleaving ending in state $s = (L,M,G)$. Then $i$ can be extended to a consistent interleaving with $s \semarrow{t}{(u, a, v)} s' = (L',M',G')$ if
 \begin{enumerate}
  \item the edge $(u,a,v)$ is in the control-flow graph for the prototype of $t$
  \item $t$ is at program node $u$ in $s$
  \item $a$ is an \emph{admissible} action for thread $t$, i.e.,
  \begin{itemize}[label={--}]
    \item if $a \equiv \code{unlock(m)}$ it must hold that $M\; \code{m} = t$,
    \item if $a \equiv \code{lock(m)}$ it must hold that $M\; \code{m} \undefined$,
    \item if $a \equiv \code{Pos(c l)}$ or $a \equiv \code{Neg(c l)}$, \code{c l} must evaluate to true, and false respectively, on the current local state \code{l} for $t$,
   \end{itemize}
  \item $s'$ \emph{reflects} the effect of action $a$ of thread $t$ (for a detailed description see
  \ifextendedversion \cref{sec:cons_details}\else the extended version of the paper\cite{arxiv-extended}\fi)
  and stores the new node $v$ for $t$.
 \end{enumerate}

\begin{toappendix}
\subsection{Details on Consistency of Interleavings}
\label{sec:cons_details}
Let $\newtid: \Threads \rightarrow \InterleavingExecutions \rightarrow \Threads$ be a function that computes the thread id for a new thread, given its parent thread id and the interleaving up to the last state before creation of the new thread.
When appending a step $\semarrow{(u,a,v)}{t} s'$ reaching state $s'$ to a consistent interleaving $i$, the resulting interleaving $i'$ is consistent only if state $s'$ \emph{reflects} the effect of the action $a$ of thread $t$. This is the case if $s' = (L', M', G')$ is as follows:

\[
  L' =
  \begin{cases}
    L \oplus \{t \mapsto (v,l)\} \oplus \{tid_c \mapsto (st_c,l_0)\} & \text{if}\ a \equiv \code{create($t_c$)}\\
    L \oplus \{t \mapsto (v,f\; l)\} & \text{if}\ a \equiv l := f\; l\\
    L \oplus \{t \mapsto (v,f\; l\; (G\; g))\} & \text{if}\ a \equiv l := f\; l\; g\\
    L \oplus \{t \mapsto (v,l)\} & \text{otherwise}\\
  \end{cases}
\]
assuming $L\; t = (u,l)$, $tid_c = \newtid(l\; \self, i)$, $st_c$ is the start node of thread template $t_c$  and $l_0 = l \oplus \{\self \mapsto tid_c \}$, and
\[
  G' =
  \begin{cases}
    G \oplus \{g \mapsto f\; l\} & \text{if}\ a \equiv g := f\; l\\
    G & \text{otherwise}\\
  \end{cases}
\]
\[
  M' =
  \begin{cases}
    M \oplus \{m \mapsto t\} & \text{if}\ a \equiv \code{lock(m)}\\
    M \oplus \{m := None\} & \text{if}\ a \equiv \code{unlock(m)}\\
    M & \text{otherwise}\\
  \end{cases}
\]

\end{toappendix}

\end{enumerate}
The set of all such consistent sequences of a program $P$ forms the set of interleaving executions $\semI{P}$. We denote the set of all consistent interleavings by $\InterleavingExecutions$.
An interleaving of Listing \ref{listing:running-example-original} is as follows:
  \begin{center}
  \ltDrawingDefs
  \scalebox{0.65}
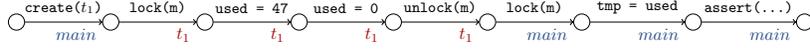
{
  \begin{tikzpicture}
    \node[programpointwide](tmpp0){};

    \programont{tmpp0}{\create{$t_1$}}{\textcolor{colormain}{main}}{tmpp1}
    \programont{tmpp1}{\lock{m}}{\textcolor{colort1}{t_1}}{tmpp2}
    \programont{tmpp2}{\code{used = 47}}{\textcolor{colort1}{t_1}}{tmpp3}
    \programont{tmpp3}{\code{used = 0}}{\textcolor{colort1}{t_1}}{tmpp4}
    \programont{tmpp4}{\unlock{m}}{\textcolor{colort1}{t_1}}{tmpp6}
    \programont{tmpp6}{\lock{m}}{\textcolor{colormain}{main}}{tmpp7}
    \programont[wide]{tmpp7}{\code{tmp = used}}{\textcolor{colormain}{main}}{tmpp8}
    \programont[wide]{tmpp8}{\assert{\dots}}{\textcolor{colormain}{main}}{tmpp9}
  \end{tikzpicture}
  }
  \end{center}

Here, the read of the global variable is extracted from the assertion, to follow the restrictions of $\Language$.
We say that an assertion is \emph{violated in an interleaving} if and only if it evaluates to false in the interleaving.
We say a program is \emph{safe w.r.t. the interleaving semantics} if and only if there are no interleavings of the program in which any of its assertions is violated.

\subsection{Thread-Modular Semantics}\label{subsec:thread-modular-semantics}
The interleaving semantics assumes the existence of a global observer that can order all events totally.
In contrast to that, the local and global trace semantics~\cite{Schwarz2021, SchwarzSSEV23} order the events of a program execution in a \emph{partial order} of local configurations of threads.
A local configuration $(n, u, l)$ consists of a local state $l$, a number $n \in \naturalnumbers$ of steps performed by the thread, and the program node $u$ the thread is at.
The local state contains the values of local variables, including the id of the thread in the variable $\self$.
This semantics requires programs from a language $\LanguageWithMg$, where there is a dedicated mutex for each global, $\Mutexes_{\Globals} = \{ m_g \mid g \in \Globals\} \subseteq \Mutexes$, and at each access to a global variable $g$ by a thread $t$, the thread holds the mutex $m_g \in \MutexesForGlobals$.
This can be achieved by e.g. inserting actions \code{lock($m_g$)} and \code{unlock($m_g$)} before and after every access to $g$.
In the partial order of a global trace, events within the same thread are ordered according to the program order.
Additionally, there exist special sets of \textit{observable} events and \textit{observer} events, which contain those events where information from one thread is published to be observed by another thread.
In our setting, \code{unlock($m$)} is an observable event with the corresponding observer being \code{lock($m$)}.
A global trace orders every observing event after the event it observes.
We refer to the reflexive transitive closure of the union of the program order, the synchronizing actions, and the create order as the \emph{causality order}.
A \emph{global trace} may contain multiple maximal elements.
We call global traces that have a unique maximal element \emph{local traces}.
The thread to which this maximal element belongs, is called the \emph{ego-thread}.
For a given ego-thread reaching some configuration, the local trace records all other local configurations that may have transitively influenced the ego-thread.

Values of global variables are not contained in a local program state; 
instead, the value of a global $g$ is determined by looking at the last write that was performed to $g$ in the local trace.
Global traces can be viewed as acyclic graphs;
for the program in \cref{listing:running-example}, the following shows an example local trace,
where mutex $m_{used}$ is abbreviated with $m_u$:
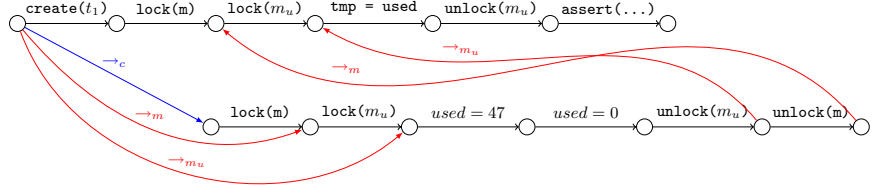
\begin{figure}
\vspace{-3em}
\ltDrawingDefs
\centering
\resizebox{0.95\linewidth}{!}{
	\begin{tikzpicture}
	\node[programpointwide](tmpp0){};

	\programon{tmpp0}{\create{$t_1$}}{tmpp1}
	\programon{tmpp1}{\lock{m}}{tmpp2}
	\programon{tmpp2}{\lock{$m_{u}$}}{tmpp3}
	\programon[wide]{tmpp3}{\code{tmp = used}}{tmpp4}
	\programon[wide]{tmpp4}{\unlock{$m_{u}$}}{tmpp5}
	\programon[wide]{tmpp5}{\assert{\dots}}{tmpp6}

	\node[programpoint,below right=of tmpp1](t1pp0){};

	\programon{t1pp0}{\lock{m}}{t1pp1}
	\programon{t1pp1}{\lock{$m_{u}$}}{t1pp2}
	\programon[wide]{t1pp2}{used = 47}{t1pp3}
	\programon[wide]{t1pp3}{used = 0}{t1pp14}


	\programon[wide]{t1pp14}{\unlock{$m_{u}$}}{t1pp15}
	\programon{t1pp15}{\unlock{m}}{t1pp16}

	\createo{tmpp0}{t1pp0}

	\mutexo[above=0mm]{tmpp0}{t1pp2}{m_{u}}{out=-70,in=220}
	\mutexo[above=0mm, pos=0.7]{t1pp15}{tmpp3}{m_{u}}{out=130,in=330}

	\mutexo{tmpp0}{t1pp1}{m}{out=-50,in=200}
	\mutexo[pos=0.8]{t1pp16}{tmpp2}{m}{out=130,in=320}
\end{tikzpicture}
}
\vspace*{-3em}
\end{figure}

Following the formal description by \citet{Schwarz2021}, requirements for a consistent global trace in particular are:
\begin{description}
  \item[Causality Order] The partial order has a unique least element.
  \item[Create Order] Each thread except the initial thread is created by exactly one create action, and each create action creates at most one thread.
  \item[Lock Order] For a given mutex $m$, each lock operation is preceded by exactly one unlock operation, or it is the first lock operation of the mutex $m$.
  Each unlock operation for $m$ is followed by at most one lock operation of $m$.
  \item[Reads of Globals] The value read from a global variable must agree with the last write performed to that variable in the global trace or the initial value.
\end{description}

\citet{Schwarz2021} give a fixpoint formulation of the set of local traces of a program and thread-modular analyses that compute abstractions of these.
We denote the set of global traces and the set of local traces of a program $P \in \LanguageWithMg$ by $\semGT{P}$ and $\semLT{P}$, respectively.
We say an assertion is \emph{violated in a local trace} if it evaluates to false in that local trace.
We say a program is \emph{safe w.r.t. the local trace semantics} if and only if there are no local traces of the program that violate any of its assertions.


\subsection{Equivalence of Interleaving Semantics and Local Trace Semantics w.r.t. Safety}
\label{subsec:equivalence-interleaving-local-traces}
We show that the interleaving and the local trace semantics agree on which programs are safe.
This later allows to show that, when a ghost witness is encoded as a program, these semantics agree on the validity of witnesses.
To formalize the correspondence between interleavings and global traces, we introduce some definitions.
The function $\threadsI$ returns the set of threads appearing in the $L$ component of the last state of an interleaving.
Similarly, we define $\threadsGT$ that returns the set of threads whose local configurations appear in a global trace.
\begin{definition}[Coincidence]\label{def:coincidence}
  Consider a program $P$ from $\LanguageWithMg$, one of its global traces $\globtrace \in \semGT{P}$ and one of its interleavings $\interleaving \in \semI{P}$.
  Let $t \in \threadsI(\interleaving) \cap \threadsGT(\globtrace)$.
  Let us consider the sequence $A_\globtrace$ of local configurations and steps of the thread $t$ in $\globtrace$ and denote it by:
  \[\sigma_0 \semarrowLT{t}{(u_0, a_0, v_0)} \sigma_1 \dots \sigma_{k-1} \semarrowLT{t}{(u_{k-1}, a_{k-1}, v_{k-1})} \sigma_k\]
  Similarly, consider the subsequence $A_{\interleaving}$ of $\interleaving$ for steps taken by $t$ and their start and target states:
  \[
    s_0 \semarrow{t}{(u'_0, a'_0, v'_0)} s_1~~s'_1 \semarrow{t}{(u'_1, a'_1, v'_1)} s_2 \dots \semarrow{t}{(u'_{l-1}, a'_{l-1}, v'_{l-1})} s_l
  \]
  By construction, for any two consecutive states $s_j, s'_j$ appearing in $A_{\interleaving}$, only threads different from $t$ may have taken steps, and thus $L_{j}(t) = L'_{j}(t)$, where $L_j$ and $L'_j$ are the first components of $s_j$ and $s'_j$, respectively.
  By projecting the states to the local variables of $t$ and its current location, adding the number of steps performed by $t$, and fusing the consecutive states with no steps in between, we obtain a sequence $A'_{\interleaving}$ of the same type as $A_\globtrace$.
  We say that the interleaving $\interleaving$ and the global trace $\globtrace$ \emph{coincide w.r.t the thread $t$}, in case that $A_\globtrace$ is equal to $A'_{\interleaving}$.
  If the same threads appear in $\interleaving$ and $\globtrace$ and they coincide w.r.t to all of these, we say they \emph{coincide}.
\end{definition}

\begin{definition}[Create-complete Global Trace]
  We say that a global trace $\globtrace$ is \emph{create-complete}, if for each created thread appearing in $\threadsGT\,{\globtrace}$, the  $\code{create}$ action that created the thread is part of $\globtrace$ and for each $\code{create}$ action, a configuration of the created thread appears.
\end{definition}

\begin{lemma}\mylabel{lemma:interleaving-global-trace-equiv}
	Let $\InterleavingExecutions = \semI{P_1}$ be the set of all interleavings for some program $P_1 \in \LanguageWithMg$,
	and let $\GlobalTraces = \semGT{P_1}$.
	\begin{enumerate}
	  \item\label{interleaving-to-globtrace}
	  For any interleaving $\interleaving \in \InterleavingExecutions$, there exists a create-complete global trace ${\globtrace \in \GlobalTraces}$ such that $\interleaving$ and $\globtrace$ coincide.
	  \item\label{globtrace-to-interleaving} For any create-complete global trace $\globtrace \in \GlobalTraces$, there exists an interleaving ${\interleaving \in \InterleavingExecutions}$ such that $\interleaving$ and $\globtrace$ coincide.
	\end{enumerate}
  \end{lemma}

\begin{proof}
  A proof can be found in \ifextendedversion \cref{sec:coincidence-proof}\else the extended version of the paper\cite{arxiv-extended}\fi.
\end{proof}

\begin{toappendix}
\subsection{Proof of \cref{lemma:interleaving-global-trace-equiv}}
\label{sec:coincidence-proof}

\begin{lemma}\mylabel{lemma:interleaving-global-trace-equiv}
	Let $\InterleavingExecutions = \semI{P_1}$ be the set of all interleavings for some program $P_1 \in \LanguageWithMg$,
	and let $\GlobalTraces = \semGT{P_1}$.
	\begin{enumerate}
	  \item\label{interleaving-to-globtrace}
	  For any interleaving $\interleaving \in \InterleavingExecutions$, there exists a create-complete global trace ${\globtrace \in \GlobalTraces}$ such that $\interleaving$ and $\globtrace$ coincide.
	  \item\label{globtrace-to-interleaving} For any create-complete global trace $\globtrace \in \GlobalTraces$, there exists an interleaving ${\interleaving \in \InterleavingExecutions}$ such that $\interleaving$ and $\globtrace$ coincide.
	\end{enumerate}
  \end{lemma}

\begin{proof}
	\begin{enumerate}
	  \item Fix an interleaving $i \in \InterleavingExecutions$ and a thread $t \in \threadsI(i)$. We obtain the global trace $\globtrace$ as follows:
	  \begin{enumerate}
		\item \label{split} For every thread $t_i \in \threadsI(i)$, including $t$, we extract the subsequence $A_{i}^{t_i}$ of $i$ that consists of steps taken by $t_i$ and their start and target states.
		As in the construction in \cref{def:coincidence}, we fuse, for every two consecutive steps in $A_{i}^{t_i}$, the target state $s_{t_i,k}$ of the first step with the start state $s'_{t_i,k}$ of the following step by mapping them to their first component $L_k$ and $L'_k$ each applied to $t_i$ which yields the coinciding reached program node and local state of $t_i$.
		Further, we add the number of steps performed by $t_i$ when reaching the state $s'_{t_i,k}$.
		This way, we obtain a sequence ${A^{t_i}_{i}}'$ for every thread $t_i$.
		\item We construct a global trace $gt$ from the partial ordering $\bigcup_{t_i \in \threadsI(i)}.\; {A^{t_i}_{i}}'$ by adding additional orderings between observable and observer events:
		\begin{itemize}
		  \item For every state $s'_{t_i,k}$ in ${A^{t_i}_{i}}'$ of some thread $t_i$ with an outgoing edge $($\code{create($t_j$)}$,t_i)$, we add $s'_{t_i,k} < s'_{t_j,0}$ to the partial ordering where $s'_{t_j,0}$ is the first state in the sequence ${A^{t_j}_{i}}'$ of the created thread $t_j$.
		  \item Further, we consider every two states $s_{k_1}$ and $s_{k_2}$ in $i$ with incoming edge $(\code{unlock(m)},t_i)$ and $(\code{lock(m)},t_j)$, respectively, where none of the edges in between are labeled with action \code{unlock(m)} or \code{lock(m)}. For those states, we pick the corresponding states $s'_{t_i,k_1}$ and $s'_{t_j,k_2}$ in ${A^{t_i}_{i}}'$ and ${A^{t_j}_{i}}'$ and add $s'_{t_i,k_1} < s'_{t_j,k_2}$ to the partial ordering.
		  \item Last, we consider all states $s_k$ in $i$ with an incoming edge $(\code{lock(m)},t_i)$ for some mutex $m$ where no state $s_{k'}$ with incoming edge $(\code{lock(m)},t_j)$ and $k' < k$ exists. Let $s'_{t,0}$ be the first state in the sequence ${A^{t}_{i}}'$. Then we add $s'_{t,0} < s_k$ to the partial ordering.
		\end{itemize}
	  \end{enumerate}
	\item
	We extend the partial order of the create-complete global trace $\globtrace$, by adding an ordering between any end node $u$ of a $\code{create}$ step and the initial node of the created thread $v$, with $u < v$.
	Then, we pick a total order on states that extends this partial order, yielding ordered states $s_0 < s_1 < \dots < s_n$, for some $n$.
	For each state except for start states of threads, there is a program action that immediately preceds it.
	Now, given this total order, we inductively construct an interleaving $\interleaving$ that coincides.
	The initial state of the interleaving is given by extending the state $s_0$ to include global variables, the state of the mutexes and of local variables of other threads, but removing the step counter of actions performed by a thread.
	Each subsequent state is computed by going to the next state in the sequence of the $(s_j)_{ 1\leq j \leq n}$ that has not been incorporated yet and which is not the start state of a thread. For such an $s_{k_0}$, we consider the action $a$ performed by thread $t$ that preceded it, and apply it to obtain the next step of the interleaving.

	The resulting interleaving coincides by construction with $\globtrace$.
	\end{enumerate}
  \end{proof}


\end{toappendix}

\begin{theorem}\label{theorem:equivalence-of-safety}
  A program $P \in \LanguageWithMg$ is safe w.r.t. interleaving semantics if and only if it is safe w.r.t. to the local trace semantics.
\end{theorem}
\begin{proof}
  We consider the following two directions:
  \begin{enumerate}
    \item \label{claim:loc-trace-to-interleaving}
    For any local trace $\loctrace \in \semLT{P}$ violating an assertion, there is an interleaving $\interleaving \in \semI{P}$ violating the same assertion.
    \item \label{claim:interleaving-to-loc-trace}
    For any interleaving $\interleaving \in \semI{P}$ violating an assertion, there is a local trace $\loctrace \in \semLT{P}$ violating the same assertion.
  \end{enumerate}
  For claim \ref{claim:loc-trace-to-interleaving}, we extend the local trace $\loctrace$ to a create-complete global trace by inserting \code{create} actions and configurations for created threads where they are still missing. With \cref{lemma:interleaving-global-trace-equiv}, it follows that there is a coinciding interleaving $\interleaving$ for this global trace.
  If the assertion $S$ at step $k$ of thread $t$ in $\loctrace$ is violated, the $k$-th step of $t$ in the interleaving violates the assertion $S$ as well, as the local state of the thread $t$ is the same in the interleaving and the global trace at the preceding state.

  For claim \ref{claim:interleaving-to-loc-trace}, for the interleaving $\interleaving$, we obtain a create-complete global trace $\globtrace$.
  Assume the assertion $S$ at the $k$-th step of thread $t$ in $\interleaving$ is violated.
  From $\globtrace$, we extract the local trace $\loctrace$ ends with the last action of $t$ in $\globtrace$.
  At the $k$-th step of thread $t$ in $\loctrace$, the assertion is violated, as the local state of the thread $t$ there is the same as in the interleaving before the $k$-th step of thread $t$.
\end{proof}

\section{Ghost Witnesses}\label{sec:ghost-witnesses}

To show the correctness of a program, we express invariants that hold at a specific location in the program, in a common format for \emph{witnesses}.
For sequential programs, this can be achieved with a mapping from locations to invariants that hold at specified location.
However, for concurrent programs, the program state may depend on the interleavings.
Therefore, in this setting, invariants should be able to reason about different interleavings.
To this end, we allow additional \emph{ghost variables} in the witness that can be also used in invariants, and \emph{ghost updates} of the corresponding variables.
Ghost updates may not modify local or global variables of the \emph{original program} $P$, but may only modify ghost local and global variables.
Thus, ghost updates may be global or local writes on ghost variables, or global reads.
First, we define witnesses:

\begin{definition}[Ghost Witness]\label{def:witness}
	A \emph{ghost witness} for a program $P$ is a tuple $(\GhostGlobalDeclarations, \GhostLocalVariables, U, I)$.
	There, $\GhostGlobalDeclarations$ is a set of \emph{ghost global declarations}, i.e., triples of the form $(name, type, value)$, for an identifier $name$, its type $type$ and a value $value$ describing the ghost variable's initial value, where newly introduced globals are disjoint from the ones existing in $P$.
	$\GhostLocalVariables$ is a set of \emph{ghost local variables} (with their types) disjoint from $\LocalVariables$.
	The \emph{ghost updates} are given by a partial function $\GhostUpdate$ from edges to (non-empty) sequences of ghost statements.
	The \emph{location invariants} are given by a partial function $\GhostInvariant$ from nodes to  boolean expression over ghost variables, global program variables, and local variables.
\end{definition}
In the following, we refer to ghost witnesses also simply as \emph{witnesses}.
We define the semantics of a witness $W$ for a program $P$ via an instrumented \emph{ghost program} $P^W_{\ghost}$ that can be obtained from the witness.
This instrumentation adds the declarations from $\GhostGlobalDeclarations$ and $\GhostLocalVariables$ and inserts the ghost updates $\GhostUpdate$ and the invariants $\GhostInvariant$ into the control flow graph.
However, for the instrumentation of the ghost updates, we combine them with statements from the original program into \emph{atomic blocks}, i.e., a sequence of multiple statements that can be only executed together, as a single atomic action.
Therefore, we first extend our programming language with atomic blocks, before we continue with the semantics of a witness.

\subsection{Atomic Blocks}
Given an arbitrary statement $a_1$ and statements that are always admissible (i.e., local update, global read, global write) $a_2, \dots, a_n$, we define an atomic block \atomic$\{ a_1; \dots; a_n\}$.
The interleaving semantics in \cref{sec:interleaving-semantics} is extended such that consistent steps may involve atomic blocks.
Such an atomic block may only be executed if the first statement $a_1$ is admissible.
The subsequent state obtained after the atomic action needs to reflect the sequential composition of the effects of the individual steps $a_1; \dots; a_n$.
Thus, an atomic block may only be executed as a whole.
In case that $a_k$, with $1 \leq k \leq n$, is an assertion, we say that $a_k$ is violated if and only if the assertion evaluates to false after executing $a_1; \dots; a_{k-1}$ on the preceding state of the interleaving.
We call the extension of the language $\Language$ with programs containing atomic blocks $\LanguageWithAtomic$.

\subsection{Witness-Instrumented Programs}

We define how, given a program $P$ and a witness $W = \Witness$, a program can be constructed from $P$ containing the instrumentation from the witness.
For a given $P$ and $W$, we define a helper function $\ghostifyInvariant$ that yields an instrumentation for asserting the invariants supplied via $\GhostInvariant$, and a function $\ghostifyUpdate$ that replaces actions on edges with atomic blocks consisting of the original action and the ghost update, if any.
More precisely,
\begin{itemize}
	\item The partial map $\ghostifyInvariant: \Nodes \partialto \Edges$ yields for a node $u$, where $\GhostInvariant(u) = i$ and $(u, a', v) \in \Edges$, for some $a'$, the edge $(u, \atomic\{a_1; \dots; a_k\}, (u,v))$, where $a_1, \dots, a_k$ is a sequence of actions to evaluate and then assert the invariant $i$. There, $(u, v)$ is a new node. If for $u$ the mapping $\GhostInvariant(u)$ is undefined, $\ghostifyInvariant\,u$ is undefined as well.
	\item The map $\ghostifyUpdate: \Edges \to \Edges$ transforms each edge $(u, a, v)$ of the original program into an edge $(n, a', v)$.
	In case that $\GhostUpdate((u, a, v))$ is specified, we set $a' = \atomic\{a; a_1; \dots; a_k\}$ where $a_1; \dots; a_k$ is the sequence of updates provided by $\GhostUpdate((u, a, v))$. In case $\GhostUpdate((u, a, v))$ is undefined, $a' = a$.
	In case $\ghostifyInvariant\,u$ is defined, we  set $n = (u, v)$, and otherwise $n = u$.
\end{itemize}

\begin{definition}[Witness-instrumented Program]\label{def:witness-instrumentation}\\
	Given a witness $W = \Witness$ for a program $P$,
	the \emph{witness-instrumented  program} $P_{\ghost}^W$ is obtained by adding $\GhostGlobalDeclarations$ to the set of global declarations and adding $\GhostLocalVariables$ to the set of local variables of $P$.
	The set of nodes in the resulting program is obtained by combining the existing set of nodes from $P$ with the new nodes added by $\ghostifyInvariant$.
	The set of edges is obtained by applying $\ghostifyInvariant$ to the nodes in $P$ and applying $\ghostifyUpdate$ to the edges of $P$.
\end{definition}
Now, validity of a witness w.r.t. the interleaving semantics can be defined:

\begin{definition}[Valid Witness w.r.t. Interleaving Semantics]\label{def:valid-witness-interleavings}
    A witness $W$ for a program $P$ is \emph{valid w.r.t. the interleaving semantics} if and only if the witness-instrumented program $P_{\ghost}^W$ is safe w.r.t. the interleaving semantics.
\end{definition}

\subsection{Preservation of Safety Properties in Ghost Programs}
Here, we go on to show that if the original program is unsafe, then a program obtained by instrumentation with a witness will be unsafe as well.
Additionally, if an assertion is safe in the original program, the same assertion will be safe in any program obtained by instrumentation with a witness.

Assume that a verifier produced the verdict that a program $P$ is safe, and produces a witness $W$ that yields the witness-instrumented program $P_{\ghost}^W$.
By checking that the witness is valid, it should follow that $P$ is safe.
In other words, violations of safety of the original program should be preserved in the witness-instrumented program.
To obtain from an interleaving of $P$ an interleaving of $P_{\ghost}^W$, the main idea is to extend the states to ghost variables, add edges introduced by $\ghostifyInvariant$ and to replace steps taken on an edge $e$ with a step on the edge introduced by $\ghostifyUpdate\,e$.
This notion is formalized by the definition of $\projectGhostify$:

\begin{definition}
Given a program $P$, a witness $W$, and an interleaving $i \in \semI{P}$, the function
$\projectGhostify$ is defined inductively and
\begin{enumerate}[label={G\arabic*}, ref={G\arabic*}]
	\item \label{ghostify:extend-initial} extends the initial state with local and global ghost variables with their initial values,
	\item \label{ghostify:replace-steps}
	for each subsequent step $(u, a, v)$ by a thread $t$ in the interleaving, adds
	\begin{itemize}
		\item two steps taken by thread $t$ corresponding to edges $e$ and $e'$, if $\ghostifyInvariant\,u = e$ and $\ghostifyUpdate\,(u, a, v) = e'$.
		\item one step taken by thread $t$ corresponding to edges $e$, if $\ghostifyInvariant\,u$ is undefined and $\ghostifyUpdate\,(u, a, v) = e$.
	\end{itemize}
\end{enumerate}
\end{definition}
As by requirement the ghost statements do not differ with regard to their treatment of non-ghost variables, the local and global states for these variables evolve in the resulting interleaving in the same manner as in the input interleaving.
We use the function $\deghostifyState$ that takes a state $(L, M, G)$, and yields a state where all local and global variables that are not defined by $P$ are removed.
How one obtains a consistent interleaving of $P_{\ghost}^W$ from a consistent interleaving of $P$ is formalized in \cref{lemma:preservation-of-runs-in-ghost-program}.

\begin{lemma}\mylabel{lemma:preservation-of-runs-in-ghost-program}
	Let $P \in \Language$ be a program and $W$ a witness for $P$. 
	Given an interleaving $i \in \semI{P}$, then $\projectGhostify\,i \in \semI{P_{\ghost}^W}$, i.e., is a consistent interleaving of $P_{\ghost}^W$.
	For the last state $s$ of $i$ and the last state $s'$ of $\projectGhostify\,i$ it holds that $\deghostifyState\,s' = s$.
\end{lemma}
The proof is by induction over the length of the interleaving $i$. A proof can be found in \cref{proof:preservation-of-runs-in-ghost-program}.
Building on \cref{lemma:preservation-of-runs-in-ghost-program}, we can now say that violations in the original program are preserved in the witness-instrumented program.

\begin{theorem}\label{theorem:violation-preservation}
	Let $P \in \Language$ be a program and $W = \Witness$ a witness for $P$.
	If there is an interleaving $i \in \semI{P}$ that violates an assertion $a$ at edge $(u, a, v)$ taken by thread $t$, the interleaving $\projectGhostify\,i \in \semI{P_{\ghost}^W}$ violates the assertion $a$
	at an edge $\ghostifyUpdate\, (u, a, v)$ taken by thread $t$.
\end{theorem}

\begin{proof}
	Let $i'$ be the prefix of $i$ that ends at the state before the assertion is violated.
	From \cref{lemma:preservation-of-runs-in-ghost-program}, it follows that $\projectGhostify\,i' \in \semI{P_{\ghost}}$, and the state of variables, owners of mutexes and locations of threads in the last state of $\projectGhostify\,i'$ agree with those in the last state of $i'$.
	In case that $\ghostifyInvariant(u, a, v)$ is undefined, the interleaving can be prolonged by a step $(u, a', v) = \ghostifyUpdate\,(u, a, v)$ taken by thread $t$, where $a'$ contains the assertion $a$ that will fail.
	If $\ghostifyInvariant(u, a, v)$ is defined, the interleaving can be prolonged by two steps $(u, a'', (u, v))$ and $((u, v), a', v)$ taken by thread $t$, where $(u, a'', (u, v)) = \ghostifyInvariant\,u$ and $((u, v), a', v) = \ghostifyUpdate\,(u, a, v)$.
	Then, the assertion $a$ is contained in $a'$ and will be violated.
\end{proof}

\begin{corollary}[Preservation of Correctness]\label{theorem:violation-preservation}
	Let $P$ be a program and $W$ be a witness for $P$. If $W$ is valid, then $P$ is safe.
\end{corollary}

A further property of interest is that the witness format ensures that assertions that cannot be violated in the original program are unaffected by the instrumentation.

\begin{theorem}\label{lemma:correctness-preservation}
	Let $W$ be a witness for a program $P \in \Language$.
	If there is an interleaving $\interleaving \in \semI{P_{\ghost}^W}$ that violates an assertion $a$ at an edge $\ghostifyUpdate\,(u, a, v)$,
	then there is an interleaving $i' \in \semI{P}$ that violates the assertion $a$ at an edge $(u, a, v)$.
\end{theorem}

A proof sketch constructing the interleaving $i'$ can be found in \ifextendedversion \cref{proof:preservation-of-correctness-in-ghost-program}\else the extended version of the paper\cite{arxiv-extended}\fi.

\begin{toappendix}
\subsection{Proofs for Preservation of Safety in Ghost Programs}
\label{app:preservation-of-safety-with-ghosts}

\subsubsection{Proof of \cref{lemma:preservation-of-runs-in-ghost-program}}
\begin{lemma}\mylabel{lemma:preservation-of-runs-in-ghost-program}
	Let $P \in \Language$ be a program and $W$ a witness for $P$. 
	Given an interleaving $i \in \semI{P}$, then $\projectGhostify\,i \in \semI{P_{\ghost}^W}$, i.e., is a consistent interleaving of $P_{\ghost}^W$.
	For the last state $s$ of $i$ and the last state $s'$ of $\projectGhostify\,i$ it holds that $\deghostifyState\,s' = s$.
\end{lemma}

\begin{proof}\label{proof:preservation-of-runs-in-ghost-program}
	We proceed by induction over the number of states $n$ of the interleaving $i$.
	For the base case $n = 1$, we have an interleaving $i \in \semI{P}$ consisting solely of the initial start state.
	Due to step \ref{ghostify:extend-initial}, $\projectGhostify\,i$ satisfies criterion \ref{interleaing:initial} for $P_{\ghost}^W$, and thus is a consistent interleaving of $P_{\ghost}^W$, and the states of the variables and mutexes in $P$ and the locations of threads agree.
	Now assume that we have an interleaving $i \in \semI{P}$ of length $n+1$.
	Then, for the interleaving $i'$ not including the last step of $i$, we have by induction hypothesis that $\projectGhostify\,i' \in \semI{P_{\ghost}^W}$, is a consistent interleaving of $P_{\ghost}^W$, and that in particular their last states agree w.r.t. to variables and mutexes in $P$ and the locations of threads.
	We show that the last steps in $\projectGhostify\,i$ not included in $\projectGhostify\,i'$ satisfy the condition \ref{interleaving:extending}.
	Assume the last step in $i$ is performed by thread $t$ at some edge $(u, a, v)$.
	We first consider the case that $\ghostifyInvariant(u, a, v)$ is undefined.
	Then, the last step in $\projectGhostify\,i$ is, according to step \ref{ghostify:replace-steps}, taken by thread $t$ for the edge $\ghostifyUpdate\,(u, a, v) = (u, a', v)$, where the effect of $a$ and $a'$ on the variables of $P$ agrees,
	and there are no additional operations in $a'$ that may not be admissible.
	As the preceding states agree, and $a$ and $a'$ do not differ w.r.t. their effect on variables in mutexes in $P$, the result states agree on the state of the mutexes and variables in $P$ and the locations of threads.
	Now, consider the case that $\ghostifyInvariant (u, a, v) = (u, a'', (u, v))$ is defined.
	Accordingly, the step $(u, a'', (u, v))$ is performed by thread $t$ first, where $a''$ is an atomic block that evaluates an invariant and then asserts it. This operation is admissible and does not affect the state of local and global non-ghost variables as well as mutexes.
	The next step inserted by $\projectGhostify\,i$ is the thread $t$ taking the edge $\projectGhostify\,(u, a,v) = ((u, v), a', v)$.
	Again, the effect of this action on the local and global variables as well as mutexes of $P$ is the same as for $(u, a, v)$.
	Thus, the last states of $\projectGhostify\,i$ and $i$ agree on the values of variables, the owners of mutexes in $P$ and the locations of threads.
\end{proof}

\subsubsection{Correctness Preservation}
\label{proof:preservation-of-correctness-in-ghost-program}
Assume we are given a program $P \in \Language$, and a witness $\Witness$ for this program, resulting in a witness-instrumented program $P_{\ghost}^W$.
Intuitively, given an interleaving $\interleaving \in \semI{P_{\ghost}^W}$ reaching a state $s$, we construct an interleaving $\interleaving' \in \semI{P}$ that reaches the state $\deghostifyState\,s$ by projecting edges that were transformed by $\ghostifyUpdate$ back.
Additionally, special care needs to be taken to cut out steps on edges introduced by $\ghostifyInvariant$.

\begin{definition}
	Let $P \in \Language$ be a program and  $W$ be a witness for $P$.
	Given an interleaving $\interleaving \in \semI{P_{\ghost}^W}$, the function $\deghostifyInterleaving$ is inductively defined and
	\begin{enumerate}[label={H\arabic*}, ref={H\arabic*}]
		\item Projects the initial state $s$ to the state $\deghostifyState\,s$;
		\item For each subsequent step $(n, a, m)$ performed by a thread $t$ yielding state $s$:
		\begin{itemize}
			\item In case that $m = (u, v)$ for some nodes $u, v$ of $P$, it does not extend its result.
			(As the action $(n, a, m)$ was inserted due to a invariant for the node $n$.)
			\item In case that $(n, a, m) = (u, a, v)$ or  $(n, a, m) = ((u, v), a, v)$ for some nodes $u, v$ of $P$, it adds one step $(u, a', v)$, where $(u, a', v) \in \Edges$ is the unique edge connecting $u$ and $v$ in $P$, and the state $\deghostifyState\,s$.
		\end{itemize}
	\end{enumerate}
\end{definition}

\begin{lemma}\label{lemma:correctness-preservation-app}
	Let $P \in \Language$ be a program and $W$ a witness for it.
	Assume there is an interleaving $\interleaving \in \semI{P_{\ghost}^W}$, that violates an assertion $a$ at an edge $\ghostifyUpdate\,(u, a, v)$.
	Then, the interleaving $\deghostifyInterleaving\,\interleaving$ is in $\semI{P}$ and violates the assertion $a$ at an edge $(u, a, v)$.
\end{lemma}

\begin{proof}[Sketch]
	One first proves by induction that interleaving $\deghostifyInterleaving\,\interleaving$ is in $\semI{P}$ and that the states obtained by the interleavings $\interleaving$ and $\deghostifyInterleaving\,\interleaving$ agree when projected to variables in $P$ and projecting away intermediate steps performed by $\interleaving$.
	Then, it follows that the assertion violated in $\interleaving$ is violated in $\deghostifyInterleaving\,\interleaving$ as well.
\end{proof}

\cref{lemma:correctness-preservation-app} provides a construction for an interleaving satisfying the property described in \cref{lemma:correctness-preservation}.
Thus, \cref{lemma:correctness-preservation} follows.
\end{toappendix}

\begin{corollary}[Preservation of Violation]\label{theorem:correctness-preservation}
	Let $P$ be a program and $W$ a witness for $P$. If $W$ is not valid because an assertion from $P$ can be violated in $P_{\ghost}^W$, then $P$ is unsafe.
\end{corollary}
\section{Valid Witnesses in the Local Trace Semantics}\label{sec:valid-witness-local-traces}
For witness-instrumented programs, we have introduced the notion of atomic blocks to the language.
While this notion is convenient for the interleaving semantics, it is less clear for the thread-modular local trace semantics, where communication between threads is assumed to happen via observable and observing events.
Thus, we discuss how atomic blocks can be encoded via lock and unlock actions, i.e., via critical sections.
Then we show the equivalence of the validity of witnesses w.r.t. the interleaving and the local traces semantics by considering programs where atomic blocks are represented in this manner.

Critical sections allow encoding atomic blocks in a way that preserves safety, as has been noted in the literature \cite{Lomet77} and used for practical implementations of atomic blocks \cite{McCloskeyZGB06}.
A witness validator therefore may choose to take either view: analyzing the ghost program w.r.t. the semantics that considers atomic blocks or w.r.t. a semantics that encodes atomic blocks as critical sections.

\setlength{\intextsep}{0pt}
\setlength{\columnsep}{7pt}

\begin{wrapfloat}{listing}{r}{0.34\textwidth}
  \setlength{\abovecaptionskip}{0em}
  \setlength{\belowcaptionskip}{1em}
  \centering
  \begin{minted}[]{c}
|\codelock{$m_{\ghost}$}|;
  |\codelock{m}|;
  |$\ghost$| = 0;
|\codeunlock{$m_{\ghost}$}|;
    \end{minted}
    \caption{Example for the encoding of an atomic block as a critical section.
    }\label{listing:mutexes-example-program}
\end{wrapfloat}

We consider the procedure $\splitAtomic$, that takes programs $P \in \LanguageWithAtomic$ with atomic blocks and yields programs in $\LanguageWithMg$ where these are encoded as critical sections.
In the resulting program, all statements accessing a global variable $g$ have to be embedded into critical sections protected by the mutex $m_g$.
Additionally, each atomic block is encoded as a critical section protected by all mutexes of the form $m_g \in \MutexesForGlobals$, if $g$ is read or written in the atomic block.
These lock-operations on mutexes inserted there are performed following some total order on $\MutexesForGlobals$.
In \cref{listing:mutexes-example-program}, one of the atomic blocks of the introductory example from \cref{listing:running-example} is encoded via critical sections.
\begin{theorem}\label{theorem:equivalence-of-atomicity-encoding}
	Let $P \in \LanguageWithAtomic$ be a program. The program $P$ is safe with respect to the interleaving semantics if and only if $\splitAtomic\,P \in \LanguageWithMg$ is safe with respect to the interleaving semantics.
\end{theorem}
A proof-sketch for the direction that if $\splitAtomic\,P$ is safe w.r.t. the interleaving semantics, also $P$ is safe w.r.t. the interleaving semantics, can be found in \ifextendedversion \cref{app:atomicity-encoding}\else the extended version of the paper\cite{arxiv-extended}\fi.
We remark that such an encoding of atomic blocks via mutexes may introduce deadlocks which do not unduly restrict the set of reachable states.
\begin{toappendix}
\subsection{Violation Preservation of Encoding of Atomicity}\label{app:atomicity-encoding}

To precisely relate interleavings of the programs $P$ and $\splitAtomic\,P$, we introduce the function $\projectSplitAtomic$.
\begin{definition}
	Given a program $P \in \LanguageWithMg$ and an interleaving $\interleaving \in \semI{P}$.
	$\projectSplitAtomic\,i$ yields the interleaving that skips steps in $i$  that lock or unlock mutexes in $\MutexesForGlobals$ and the result states of these steps, and where occurrences of mutexes in $\MutexesForGlobals$ are removed from all states.
\end{definition}

\begin{theorem}\label{theorem:split}
	Given a program $P \in \LanguageWithAtomic$, and an interleaving $i \in \semI{P}$.
	There is an interleaving $i'$ for the program $\splitAtomic\,P$, such that the global state, the owners of mutexes and the local states of all threads of the state attained at every index in $i$ and $\projectSplitAtomic\,i'$ are the same.
\end{theorem}

\label{sec:proof-split}
\begin{proof}
	For a given $P$ and a program $\splitAtomic\,P$, the function $\extend$ that, given $i$, constructs $i'$, is defined inductively. Given an interleaving consisting solely of the start state, it returns the interleaving consisting of the start state where the mutexes in $\MutexesForGlobals$ are added.
	For an interleaving $i$ consisting of $n+1$ steps, it computes the interleaving $\extend\, i'$, where $i'$ is the prefix of $i$ containing the first $n$ steps, and prolongs the result in the following manner.
	In case the last step taken in $i$ is a base statement $a$, and a global $g$ is accessed in $a$,
	the interleaving $\extend\, i'$ is prolonged by a sequence $\code{lock}(m_g); a; \code{unlock}(m_g)$, and the state is updated accordingly in the resulting steps.
	In case the last step taken in $i$ is an atomic statement $a'$, and the global variables being accessed in $a'$ are $g_1, \dots, g_k$, the interleaving $\extend\, i'$ is prolonged by a sequence $\code{lock}(m_{g_1});\dots; \code{lock}(m_{g_k}); a; \code{unlock}(m_{g_k});\dots; \code{unlock}(m_{g_1})$.
	Here, it is assumed without loss of generality, that the mutexes $m_{g_1}, \dots, m_{g_k}$ are in the total order defined for the mutexes in $\MutexesForGlobals$ in $\splitAtomic\,P$.
	For every interleaving returned by $\extend$, the set of mutexes in $\MutexesForGlobals$ that is held by any thread is empty at the end of the interleaving, ensuring that prolonging it with \code{lock}($m_g$) operations is possible.
\end{proof}

\begin{corollary}\label{cor:violation-preservation}
	Let $P \in \LanguageWithAtomic$ be a program.
	If there is an interleaving $i \in \semI{\splitAtomic\,P}$ where an assertion is violated at step $k$,
	there is a trace $i' \in \semI{P}$, such that the assertion is violated at the step with index $k$ of $\projectSplitAtomic\,i'$.
\end{corollary}

According to \cref{cor:violation-preservation}, if there is any interleaving on a program using atomic blocks violating a given assertion, there is also a violating interleaving on the program where the atomic block is encoded as a critical section.
Thus, the encoding of a ghost program via critical sections can be used to soundly verify that there are no assertion violations in the ghost program, and therefore no assertion violations in the original program.

\begin{theorem}
	Let $P \in \LanguageWithAtomic$ be a program.
	If there is an interleaving $\interleaving' \in \semI{\splitAtomic\,P}$ where an assertion is violated in $\projectSplitAtomic\,i'$,
	there is an interleaving $\interleaving \in \semI{P}$, such that the same assertion is violated.
\end{theorem}
\begin{proof}
	Without proof.
\end{proof}



\end{toappendix}
Using this encoding, we can define the validity of a witness with respect to the local trace semantics.
\begin{definition}[Valid Witness w.r.t. Local Trace Semantics]\label{def:valid-witness-local-traces}
    A witness $W$ for a program $P \in \Language$, resulting in the witness-instrumented program $P_{\ghost}^W$, is \emph{valid w.r.t. the local trace semantics} if and only if the program  $\splitAtomic\,P_{\ghost}^W$ is safe w.r.t. the local trace semantics.
\end{definition}

Finally, we can relate this notion back to the validity of witnesses in the interleaving semantics.
With the
definition of a valid witness w.r.t. the local trace semantics referring back to safety and by
encoding atomicity in witness-instrumented programs via critical sections, we can lift the agreement of the two semantics on safety to an agreement on validity of ghost witnesses.
\begin{theorem}\label{theorem:equivalence-of-validity}
	Let $W$ be a witness for a program $P \in \Language$.
	The witness is valid with respect to the interleaving semantics if and only if it is valid with respect to the local trace semantics.
\end{theorem}

\begin{proof}
	By \cref{def:valid-witness-interleavings}, the witness $W$ is valid with respect to the interleaving semantics if and only if the witness-instrumented program $P_{\ghost}^W$ is safe.
	According to \cref{theorem:equivalence-of-atomicity-encoding}, the program $P_{\ghost}^W$ is safe with respect to the interleaving semantics if and only if the program $\splitAtomic\,P_{\ghost}^W$ is safe with respect to the interleaving semantics.
	As $\splitAtomic\,P_{\ghost}^W \in \LanguageWithMg$, according to \cref{theorem:equivalence-of-safety}, the program $\splitAtomic\,P_{\ghost}^W$ is safe with respect to the interleaving semantics if and only if it is safe with respect to the local trace semantics.
	With \cref{def:valid-witness-local-traces}, the statement follows.
\end{proof}

\section{Witness Generation with Thread-Modular Abstract Interpretation}
\label{sec:generation}
Here, we outline how the analysis results of two thread-modular analyses based on abstract interpretation can naturally be expressed using ghosts.
While both analyses~\cite{Schwarz2021,SchwarzSSEV23} are implemented in \Goblint and proven correct relative to the local trace semantics, their abstractions differ:
one computes relational invariants per mutex, while the other computes non-relational invariants per global.

The relational \emph{mutex-meet} analysis by \citet{SchwarzSSEV23} assumes that each global variable~$g$ is (write\nobreakdash-)protected by a set of mutexes $\mathcal{M}[g]$ which are held at every write to $g$.
For every mutex $m$, the analysis computes a \emph{mutex invariant}~$[m]$ containing relational information between those global variables $\mathcal{G}[m] = \{g \mid m \in \mathcal{M}[g]\}$ only written when $m$ is held.
The computed invariant holds whenever no thread holds the mutex~$m$, i.e., it may be violated while some thread has exclusive access, but holds again once $m$ is released.
To express mutex invariants, we introduce for each mutex $m$ a boolean ghost variable~$m_{\mathrm{locked}}$, indicating whether any thread has locked $m$.
The variable is initialized to $\mathsf{false}$ and corresponding ghost variable updates are added to the witness:
every lock of $m$ is instrumented with ${m_{\mathrm{locked}} = \mathsf{true}}$ and every unlock with ${m_{\mathrm{locked}} = \mathsf{false}}$.
This allows every mutex invariant $[m]$ to be added to the witness as
\(
    {\lnot m_{\mathrm{locked}} \implies [m]} 
\).
Listing \ref{listing:running-example-annotated} is an example obtained by such an instrumentation.

\Goblint additionally distinguishes two phases of the program: when it is single-threaded and when it has become multithreaded.
Thread-modular analysis are only used for the latter,
while, in the former phase, global variables are analyzed flow-sensitively to, e.g., retain precision during the initialization phase of the program.
We introduce another boolean ghost variable $\multithreaded$, indicating whether the program has become multithreaded.
It is initialized to $\mathsf{false}$ and all thread creation actions that may potentially create the first additional thread are instrumented with $\multithreaded = \mathsf{true}$.
Thus, every invariant~$[m]$ is actually added as
\(
    (\multithreaded \land \lnot m_{\mathrm{locked}}) \implies [m] 
\).

The non-relational \emph{protection} analysis \cite{Schwarz2021} computes for each global
variable~$g$ a \emph{protected invariant}~$[g]$ which describes values of $g$ when no thread holds a mutex from $\mathcal{M}[g]$.
Every protected invariant $[g]$ is added to the witness as
\(
    (\multithreaded \land \bigwedge_{m \in \mathcal{M}[g]} \lnot m_{\mathrm{locked}}) \implies [g]
\).

The analyses and their computed invariants are flow-insensitive on global variables, so the invariants should be valid at all program locations that fulfill certain conditions.
Repeating the flow-insensitive invariants at every location is
costly for validators; instead, we add flow-insensitive invariants after each thread create in the program.

\section{Ghost Witnesses for C Programs}\label{sec:format}
Here, we propose a uniform format for correctness witnesses with ghost variables (based on \cref{def:witness})
for concurrent C programs.
The format enables witness exchange between verifiers and witness validators.
Our format is compatible with an existing format~\cite{Ayaziova2024} for correctness witnesses of sequential programs
utilized by the community around the International Competition on Software Verification (SV-COMP)~\cite{Beyer2024}
and is realized in YAML~\cite{YamlSpec21}.

A witness in this format contains invariants and specifies at which locations of the C~program these invariants hold.
Rather than referring to a specific control-flow graph, locations are expressed by a line number and a column number in the source C file.
Corresponding to the elements of the tuple $(\GhostGlobalDeclarations, \GhostLocalVariables, U, I)$ from \cref{def:witness},
the format defines the following types of entries in the witness:
\begin{description}
	\item[\ttbox{ghost\_variables}] A ghost global variable declaration consists of a primitive C~type, a unique identifier, and an initial value given as a (side effect-free) C~expression over global variables of the program.
    \item[\ttbox{ghost\_updates}] A ghost update consists of an identifier (referring to a ghost variable declaration), a location at which the update is performed, and a (side effect-free) C~expression over program variables (which must be in scope at the location) and ghost variables.
	\item[\ttbox{location\_invariant}] A location invariant consists of a location and an invariant represented as a (side effect-free) C~expression over program variables (which must be in scope at the location) and ghost variables.
\end{description}
A more detailed description of these entries can be found in \ifextendedversion \cref{tab:structure-ghost-instrumentation-entry} in the appendix\else the extended version of the paper\cite{arxiv-extended}\fi.
We omit the definition of local ghost variables, as in C the reading from and writing to global variables does not require local variables as in our formalization.
A full YAML~schema for witnesses is available~\cite{GhostWitnessFormat24} and documented~\cite{GhostWitnessDocumentation24}.
\Cref{lst:loc-invariant,lst:ghost-decl-update} show two excerpts of a correctness witness for the program in \cref{e:running-example}, the full witness can also be found \ifextendedversion in the appendix (\cref{lst:witness-example})\else the extended version of the paper\cite{arxiv-extended}\fi.

\begin{toappendix}
\begin{table}[h]
  \centering
  \caption{Structure of a \ttbox{ghost\_instrumentation} witness entry composed of declarations for \ttbox{ghost\_variables} and \ttbox{ghost\_updates} to alter the values of the variables.}
  \renewcommand{\arraystretch}{1.05}
\setlength{\tabcolsep}{3.5pt}
\setlength{\multilen}{6.5cm}
\begin{tabular}{lll}
  \toprule
  Key & Value & Description \\
  \midrule
  \ttbox{ghost\_variables} & array                 & the list of \ttbox{ghost\_variable} items; see below \\
  \ttbox{ghost\_updates}   & array                 & the list of \ttbox{ghost\_update} items; see below \\
  \separator{content of \ttbox{ghost\_variable}} \\
  \ttbox{name}             & scalar                & the name of the ghost variable \\
  \ttbox{type}             & scalar                & the type of the ghost variable \\
  \ttbox{scope}            & \ttbox{global}        & the scope of the variable is always \ttbox{global} \\
  \ttbox{initial}          & mapping               & the initial value of the variable; see below \\
  \separator{content of \ttbox{initial}} \\
  \ttbox{value}            & scalar                & the initial value of the ghost variable \\
  \ttbox{format}           & \ttbox{c\_expression} & the format is always \ttbox{c\_expression} \\
  \separator{content of \ttbox{ghost\_update}} \\
  \ttbox{location}         & mapping               & the location of the update as structured in~\cite{Ayaziova2024} \\
  \ttbox{updates}          & array                 & the list of \ttbox{update} items; see below \\
  \separator{content of \ttbox{update}} \\
  \ttbox{variable}         & scalar                & the name of the ghost variable \\
  \ttbox{value}            & scalar                & the value assigned to the ghost variable \\
  \ttbox{format}           & \ttbox{c\_expression} & the format is always \ttbox{c\_expression} \\
  \bottomrule
\end{tabular}

  \label{tab:structure-ghost-instrumentation-entry}
\end{table}

\end{toappendix}

\begin{toappendix}
\begin{listing}[htbp]
  \begin{minted}{yaml}
- entry_type: invariant_set
  metadata: |\dots|
  content:
    - invariant:
        type: location_invariant
        location: { line: 4, |\dots| }
        value: |$\ghost$| == 0 |$\Longrightarrow$| used == 0
        format: c_expression
- entry_type: ghost_instrumentation
  metadata: |\dots|
  content:
    ghost_variables:
      - name: |$\ghost$|
        type: int
        scope: global
        initial:
          value: 0
          format: c_expression
    ghost_updates:
      - location: { line: 4, |\dots| }
        updates:
          - variable: |$\ghost$|
            value: 1
            format: c_expression
      - location: { line: 6, |\dots| }
        updates:
          - variable: |$\ghost$|
            value: 0
            format: c_expression
      - location: { line: 8, |\dots| }
        updates:
          - variable: |$\ghost$|
            value: 1
            format: c_expression
      - location: { line: 11, |\dots| }
        updates:
          - variable: |$\ghost$|
            value: 0
            format: c_expression
\end{minted}

  \caption{Witness with ghosts in the YAML-format for the program from \cref{listing:running-example}.}
  \label{lst:witness-example}
\end{listing}
\end{toappendix}

\setlength{\intextsep}{0pt}
\setlength{\columnsep}{0pt}
\begin{wrapfloat}{listing}{r}{0.47\textwidth}
\setlength{\abovecaptionskip}{0em}
\setlength{\belowcaptionskip}{0.5\baselineskip}
\centering
\begin{minted}[fontsize=\scriptsize]{yaml}
- entry_type: invariant_set
  metadata: |\dots|
  content:
    - invariant:
        type: location_invariant
        location: { line: 4, |\dots| }
        value: |$\ghost$| == 0 |$\Longrightarrow$| used == 0
        format: c_expression
\end{minted}
\caption{Location invariant.}
\label{lst:loc-invariant}
\end{wrapfloat}
Before defining \emph{validity} for these witnesses,
we briefly discuss the semantics of concurrent C.
Notably, the semantics of concurrent C~programs is not based on the interleaving of well-defined atomic steps as in our formal model;
the granularity of steps is implementation-defined.
For instance, depending on compiler and target platform,
a write of a \qty{64}{\bit} shared variable may be performed atomically,
or it may be split into two write accesses of \qty{32}{\bit}, or even more granularly.
Since actions of other threads may interleave arbitrarily with these accesses,
leading to unexpected results,
C forbids (non-atomic) concurrent accesses to shared variable,
and declares such \emph{data races} as undefined behaviour.
As we are concerned with proving correctness (wrt.\ reachability properties) on C~programs,
we only consider well-behaved C~programs without such data races.
Moreover, C allows atomic operations on shared data to specify different \emph{memory orders},
allowing for certain weak memory models to be used. 
We do not support these memory orders,
and always assume \emph{sequential consistency}~\cite{Lamport79}.

\begin{wrapfloat}{listing}{r}{0.47\textwidth}
\setlength{\abovecaptionskip}{0em}
\setlength{\belowcaptionskip}{0.5\baselineskip}
\vspace{-0.5\baselineskip}
\centering
\begin{minted}[fontsize=\scriptsize]{yaml}
- entry_type: ghost_instrumentation
  metadata: |\dots|
  content:
    ghost_variables:
      - name: |$\ghost$|
        type: int
        scope: global
        initial:
          value: 0
          format: c_expression
    ghost_updates:
      - location: { line: 4, |\dots| }
        updates:
          - variable: |$\ghost$|
            value: 1
            format: c_expression
\end{minted}
\caption{Ghost declaration and update.}
\label{lst:ghost-decl-update}
\end{wrapfloat}
The validity of our witnesses for C~programs is based on \cref{def:valid-witness-interleavings}.
The witness format for sequential programs by Ayaziov{\'{a}} et~al.~\cite{Ayaziova2024} specifies that
each \ttbox{location\_invariant} must hold immediately before executing
the statement or declaration it is attached to,
i.e., the statement or declaration at the specified location.
In concurrent programs we must additionally define which global states the evaluation of an invariant may observe.
In particular, as the granularity of write accesses by other threads to shared variables is unclear,
an invariant could potentially observe arbitrary intermediate states of such a write.
This is problematic for witnesses, as they cannot express useful information over the global state in that case.
To address this issue, we use the notion of \emph{sequence points} in C, a set of well-defined points in an execution where all side effects of the previous actions are guaranteed to have been performed and no side effects of subsequent actions have taken place~\cite{ISO:C99}.
For validity of a witness in our format, it suffices that every location invariant in the witness holds whenever one thread is at the specified location and every other thread is at a sequence point.
Additionally, each location invariant is evaluated atomically, i.e., no other thread can interleave while the location invariant is evaluated.

The effect of the declaration of a ghost variable $g$ with type $T$ and initial value $v$
corresponds to inserting a global declaration \code{T g = v;} in the C~program.
The declaration is executed
after the global variables of the program are declared.
As a result, the evaluation of $v$ can access
the initial value of the program's global variables.
The update of a ghost variable $g$ with the value $v$ has the same semantics as the C~statement \code{g = v;}.
As described in \cref{def:witness-instrumentation}, the update is executed atomically with the action at the given witness location.
Unlike the witnesses in the control flow graph as in \cref{def:witness}, we cannot in general attach the ghost update to an arbitrary statement, as we cannot assume that this statement is executed atomically (e.g., it may contain a nested function call or might execute other side-effects).
Therefore, we allow ghost updates only at locations that point to the beginning of one of the following statements.

\begin{toappendix}
\begin{table}[h]
\setlength\tabcolsep{1em}
\centering
\caption{Supported \code{pthread} functions for ghost updates with the corresponding location}
\label{tab:pthread-functions}
\begin{tabular}{l|l}
\toprule
Function name                  & Location of the ghost update     \\
\midrule
\code{pthread\_create}         & thread creation                  \\
\code{pthread\_mutex\_lock}    & locking of the mutex             \\
\code{pthread\_mutex\_unlock}  & unlocking of the mutex           \\
\code{pthread\_rwlock\_rdlock} & locking of the read-write lock   \\
\code{pthread\_rwlock\_wrlock} & locking of the read-write lock   \\
\code{pthread\_rwlock\_unlock} & unlocking of the read-write lock \\
\code{pthread\_cond\_wait}     & locking of the mutex             \\
\bottomrule
\end{tabular}
\end{table}
\end{toappendix}

\begin{itemize}
	\item
	If the statement is a call to a function from the pthread library (e.g. for thread-creation or locking of a mutex; the full list can be found in \cref{tab:pthread-functions}),
  then the ghost update is performed at the sequence point after evaluating the arguments of the function and together with the action leading to the next sequence point (e.g., locking of the mutex for \code{pthread\_mutex\_lock}, thread creation for \code{pthread\_create}).
	\item
	If the statement is an assignment, the ghost update is performed after the assignment, and it happens atomically with the actual write (i.e., only the write itself is performed atomically, not the evaluation of the expressions).

	Performing the write atomically may reduce the possible interleavings,
	but for programs without data races the behaviour is equivalent.
\end{itemize}

In the benchmarks used by SV-COMP, there are two additional functions that allow one to define an atomic block in C~\cite{SVCompRules24}.
The verifiers are instructed to assume that the code between \code{\_\_VERIFIER\_atomic\_begin} and \code{\_\_VERIFIER\_atomic\_end} is executed in one atomic step.
As we aim for our witnesses to be used by the SV-COMP community,
we also allow ghost updates at the call of these two functions.
For a call to \code{\_\_VERIFIER\_atomic\_begin} (resp. \code{\_\_VERIFIER\_atomic\_end}), the update is performed after (resp. before) the call to this additional function,
i.e., the update is performed atomically with the code between these calls.
\section{Validation with Software Model Checking}\label{sec:validation-with-model-checking}

Witnesses expressed in the format described in the previous section can be exchanged with, and validated by, other verification tools,
to confirm the results of the verification using tools with a significant technological and/or conceptual difference,
and thereby increase trust in the verification tool or potentially uncover bugs in the verification.
To demonstrate this, we implemented witness validation in the software model checking tool \UltimateGemCutter%
~\cite{GemCutterSVCOMP,GemCutterPLDI},
which is developed as part of the \Ultimate program analysis framework~\cite{ultimate}. 

\UltimateGemCutter is based on an interleaving semantics (\cref{sec:interleaving-semantics}).
It uses a \emph{commutativity relation} between statements to identify pairs of statements (of different threads) whose relative order does not affect program correctness.
In its analysis, \GemCutter groups interleavings into equivalence classes based on the induced \emph{Mazurkiewicz equivalence}~\cite{Mazurkiewicz86},
i.e., it considers interleavings as equivalent if they only differ in the order of commuting statements.
\GemCutter proves correctness for one representative per equivalence class,
using the \emph{trace abstraction} algorithm~\cite{Heizmann09TraceAbstraction}, a counterexample-guided abstraction refinement scheme.
If the proof for the representatives succeeds, \GemCutter soundly concludes correctness of the entire program.
To validate witnesses, \GemCutter instruments the given program with the ghost code and invariants from the witness (similar to \cref{def:witness-instrumentation}),
and applies its usual verification algorithm to the instrumented program.

\begin{table}[t]
	\centering
	\caption{Number of \emph{confirmed} witnesses. In \emph{witness confirmation} mode, \GemCutter checks validity of witness invariants while ignoring program correctness.}
	\label{tab:eval-validate}
		\setlength{\tabcolsep}{6pt}
	\begin{tabular}{lrrrr}
	    \toprule
		witnesses   & protection-$\ghost$ & mutexmeet-$\ghost$ & protection & mutexmeet \\ \midrule
		\multicolumn{5}{c}{witnesses for correct programs}                                          \\ \midrule
		total       &                 181 &                217 &              159 &             230 \\
		confirmed   &                 171 &                193 &              122 &             167 \\
		rejected    &                   0 &                  0 &                0 &               0 \\
		out of resources &              5 &                 19 &               34 &              56 \\ \midrule
		\multicolumn{5}{c}{witnesses for incorrect programs}                                        \\ \midrule
		total       &                 282 &                288 &              295 &             300 \\
		confirmed   &                 192 &                164 &              224 &             130 \\
		rejected    &                   0 &                  0 &                0 &               0 \\
		out of resources &             85 &                117 &               68 &             163 \\ \bottomrule
	\end{tabular}
\end{table}

\subsection{Experimental Setup}
Complementary to the theoretical equivalence results of \cref{theorem:equivalence-of-validity},
we empirically evaluate the impact of witness validation regarding the following questions:
\begin{description}
\item[Q1] Can witnesses generated by a verifier be validated by a different tool based on different semantics, analysis approaches, and technological foundations?
\item[Q2] What is the overhead of witness validation? Can witnesses help to accelerate the analysis, compared to verification without a witness?
\item[Q3] Can witness exchange and validation with tools based on a different approach help to find bugs in verification tools?
\end{description}

For our evaluation, we analysed the programs in the \textit{ConcurrencySafety-Main} category of SV-COMP~\cite{Beyer2024} with \Goblint and generated correctness witnesses.
For each of the two different analyses (\emph{mutex-meet} and \emph{protection}) described in \cref{sec:generation}, we generated witnesses with and without ghost variables,
yielding four witness benchmark sets.
The witnesses without ghosts simply contain location invariants within critical sections, immediately after lock operations.
The non-ghost invariants only refer to variables protected by mutexes held at their locations, implicitly using locations to encode concurrency information.
Thus, the number of invariants, their locations and the invariant expressions themselves are incomparable between the two kinds of witnesses.
We excluded witnesses that do not contain any (non-trivial) invariants.
Note that, as an abstract interpreter, \Goblint cannot prove incorrectness of programs, but is still able to generate valid invariants for incorrect programs (though they are, naturally, insufficient to prove correctness).
Hence, the generated witness benchmark sets include witnesses for incorrect programs, where the invariants reflect information that \Goblint has derived about the program, and can be checked by our \emph{witness confirmation} analysis.

We validated these witnesses using \GemCutter.
The validation was executed on an AMD Ryzen Threadripper 3970X at \qty{3.7}{\giga\hertz} with a time limit of \qty{900}{\second}, a memory limit of \qty{16}{\giga\byte} and a CPU core limit of \num{2}.
For reliable measurements, the experiments were carried out using the \textsc{BenchExec} framework~\cite{BenchExec}.
An artifact allowing reproduction of our experiments is available~\cite{artifact}.

\begin{table}[t]
	\centering
	\caption{Validated witnesses and required time.
	  The second line indicates verification results on those benchmarks where a witness in the corresponding witness set existed.
    }
	\newcommand\timehead{\begin{tabular}{@{}c@{}}time \\ h:m:s\end{tabular}}
	\label{tab:eval-verify}
		\setlength{\tabcolsep}{7pt}
	\begin{tabular}{lrrrrrrrr}
	    \toprule
		witnesses & \multicolumn{2}{c}{protection-$\ghost$} & \multicolumn{2}{c}{mutexmeet-$\ghost$} & \multicolumn{2}{c}{protection} & \multicolumn{2}{c}{mutexmeet} \\
		\cmidrule(lr){2-3} \cmidrule(lr){4-5} \cmidrule(lr){6-7} \cmidrule(lr){8-9}
		          &           \# &                \timehead &           \# &                \timehead &       \# &                \timehead &       \# &               \timehead  \\
	    \midrule
%
		validation &          119 &          \Duration{5493} &          123 &         \Duration{5367} &       100 &          \Duration{7289} &      110 &          \Duration{9432} \\
		verification &        120 &          \Duration{4047} &          131 &         \Duration{4709} &       101 &          \Duration{2674} &      118 &          \Duration{4063} \\
		\bottomrule
	\end{tabular}
\end{table}

\subsection{Results}
\cref{tab:eval-validate} shows the number of witnesses in each witness benchmark set, and how many of them are \emph{confirmed} resp.\ rejected by \GemCutter.
In the \emph{witness confirmation} mode used here, \GemCutter only checks if the asserts of the instrumented program that correspond to witness invariants can be violated,
while ignoring the asserts already present in the original program.

Despite the significant differences between \Goblint and \GemCutter, the communication via the witness format is successful (\textbf{Q1}).
The analysed witnesses contain non-trivial information:
on average, witnesses in the set \mbox{``protection-$\ghost$''} (resp.\ ``mutexmeet-$\ghost$'') contain $44.1$~invariant entries (resp.\ $43.9$ ) and $3.8$~ghost updates (resp.\ $3.6$),
whereas the witnesses in ``protection'' (resp.\ ``mutexmeet'') contain on average 179 (resp.\ 749.9) invariant entries.
For witnesses that are not confirmed, \GemCutter either runs out of time or memory (see \cref{tab:eval-validate}),
or crashes due to ghost updates at unsupported locations (12~cases)
or unsupported C~features in the program or witness (58~cases).

\Cref{tab:eval-verify} shows the results of running \GemCutter in full \emph{witness validation} mode (the mode used in SV-COMP),
where the witness invariants must be confirmed and correctness of the program must be proven.
As we are working with \emph{correctness} witnesses, we only consider correct programs for the validation.
For \Goblint witnesses validated by \GemCutter,
witness validation generally incurs an overhead rather than accelerating the verification~(\textbf{Q2}).
However, for almost all programs verified by \GemCutter, the corresponding witnesses can also be validated with a moderate slowdown.
In two cases, validation of a witness from the ``mutexmeet-$\ghost$'' resp.\ ``protection-$\ghost$'' set succeeds (in less than \qty{250}{\second}) while the verification of the corresponding program times out.

Over the duration of this work, its implementation and (re-)evaluation revealed 29~bugs in \Goblint and \GemCutter and lead to their fixing in most cases (\textbf{Q3}).
These bugs already existed prior to and independently of this work,
and crucially, they were revealed on a subset of the very same SV-COMP benchmarks that have been used over the years to test and evaluate these tools.
A full overview of these bugs is provided in \ifextendedversion \cref{sec:bugs}\else the extended version of the paper\cite{arxiv-extended}\fi.
In particular, four soundness-critical bugs were discovered in the verification performed by \Goblint, and 13 bugs in \Goblint's witness production (including missing or incorrect invariants).
In \Ultimate, five issues concerned the support for ACSL specifications (which is reused as part of the witness instrumentation), two bugs concerned the translation of C code to \Ultimate's internal representation, one bug concerned the internal representation, and one bug was specific to witness validation.
Notably, the bugs regarding ACSL support and translation of C~code are soundness-critical,
and cannot be caught by correctness checks inside \Ultimate that operate on the internal representation.
Thus, the witness exchange and validation significantly improved both tools.

\paragraph*{Threats to validity.}
We conducted experiments with \Goblint\ and \GemCutter, producing ghost witnesses in \Goblint\ (albeit for two different analysis approaches)
and validating them in \GemCutter.
While further experiments with more analyzers would offer additional evidence of the suitability of the format, \Goblint\ and \GemCutter\
employ radically different approaches and technologies; thus, we consider our experiments strong evidence that the format succeeds in
exchanging information across such divides.
We do not study the exchange of ghost witnesses in the opposite direction here.
Lastly, the evaluation was performed on tasks from the ConcurrencySafety-Main category of SV-COMP, implying the usual concerns about generalizability
to real-world programs of any experiment performed on SV-COMP benchmarks also apply here.

\begin{toappendix}
\label{sec:bugs}

Over the duration of this work, its implementation and (re-)evaluation revealed numerous bugs in the two tools and lead to their fixing in most cases.
These already existed prior to and independently from this work.
Thus, the collaboration and witness exchange improved both tools as a whole.

\newcommand{\goblintIssue}[1]{\href{https://github.com/goblint/analyzer/issues/#1}{issue \##1}}
\newcommand{\goblintCilIssue}[1]{\href{https://github.com/goblint/cil/issues/#1}{\textsc{CIL} issue \##1}}
\newcommand{\goblintPR}[1]{\href{https://github.com/goblint/analyzer/pull/#1}{PR \##1}}
\newcommand{\goblintCilPR}[1]{\href{https://github.com/goblint/cil/pull/#1}{\textsc{CIL} PR \##1}}
\newcommand{\UltimateCommit}[1]{\href{https://github.com/ultimate-pa/ultimate/commit/#1}{#1}}

\paragraph{\Goblint.}
Fixed:
\begin{itemize}
    \item Verification:
    \begin{enumerate}
        \item Unsound invariants from interval set domain widening (\goblintIssue{1356}, \goblintIssue{1473}, \goblintPR{1476}).
        \item Unsound relational analysis due to \texttt{extern} variables (\goblintIssue{1440}, \goblintPR{1444}).
        \item Unsound relational analysis due to \texttt{\_\_VERIFIER\_atomic} mutex (\goblintIssue{1440}, \goblintPR{1441}).
        \item \texttt{continue} handling in syntactic loop unrolling (\goblintPR{1369}).
    \end{enumerate}
    \item SV-COMP witness generation:
    \begin{enumerate}
        \item Logical expression intermediate locations (\goblintIssue{1356}, \goblintCilPR{166}).
        \item Multiple/composite declaration initializer intermediate locations (\goblintCilPR{167}, \goblintPR{1372}).
        \item \texttt{for} loop initializer \texttt{location\_invariant} location (\goblintCilPR{167}).
        \item \texttt{for} loop \texttt{loop\_invariant} location (\goblintIssue{1355}, \goblintPR{1372}).
        \item No \texttt{loop\_invariant} for \texttt{goto} loop (\goblintPR{1372}).
        \item \texttt{do}-\texttt{while} loop \texttt{loop\_invariant} location (\goblintPR{1372}).
        \item Exclude internal struct names from generated invariants (\goblintPR{1375}).
        \item Improve invariants for syntactically unrolled loops (\goblintPR{1403}).
        \item Exclude trivial \texttt{\_Bool} invariants (\goblintIssue{1356}, \goblintPR{1436}).
        \item Exclude trivial congruence invariants (\goblintIssue{1218}).
        \item Exclude invariants with out-of-scope local variables (\goblintIssue{1361}, \goblintPR{1362}).
        \item Mutex-meet invariants do not account for initial values (\goblintIssue{1356}).
        \item Exclude \Goblint stubs from witnesses (\goblintPR{1334}).
    \end{enumerate}
\end{itemize}
Not fixed:
\begin{enumerate}
    \item Witness invariant parser does not handle \texttt{typedef} (\goblintCilIssue{159}).
    \item Autotuner enables all integer domains (\goblintIssue{1472}).
    \item Cannot have \texttt{location\_invariant} before loop (\goblintIssue{1391}).
\end{enumerate}

\paragraph{\Ultimate.}
Fixed:
\begin{enumerate}
    \item (Validation) Location invariants at labels were not handled properly (\UltimateCommit{d55e39c})
    \item (C-Translation) Model return value of \texttt{\_\_VERIFIER\_nondet\_bool} properly (\UltimateCommit{f7d84c9})
    \item (ACSL) Enhance grammar for variable names (\UltimateCommit{b9e0ff7}, \UltimateCommit{ab2e0ac})
    \item (ACSL) Support \texttt{\_Bool} type (\UltimateCommit{9435525}, \UltimateCommit{cb09d65})
    \item (ACSL) Support for \texttt{\&}(\UltimateCommit{ed2a5ba})
    \item (ACSL) Fix \texttt{->} translation (\UltimateCommit{6f5224f})
    \item (CFG) Fix support for nested atomic blocks (\UltimateCommit{a38a16e})
\end{enumerate}

Not fixed:
\begin{enumerate}
    \item (C-Translation) Missing support for \texttt{pthread}-attributes
    \item (ACSL) Missing support for function pointers
\end{enumerate}
\end{toappendix}
\section{Related Work}\label{sec:related-work}


The exchange of information between analyzers through an analysis-agnostic exchange format
is actively researched and has led to the development of cooperative verification~\cite{CoVeriTeam2022,HaltermannWehrheim2022}.
Witnesses are central to SV-COMP with many tools specifically developing
witness generation and validation approaches~\cite{saan_correctness_2024,BeyerF20,SvejdaBK20,MilaneseM24}.
The first generation relied on automata exported as GraphML~\cite{Beyer2015,Beyer2016,BeyerF20},
which has been deprecated in favor of a YAML-based format~\cite{Ayaziova2024}.
There are other means of exchanging analysis information that focus on external results,
such as SARIF (Static Analysis Results Interchange Format)~\cite{OASISSarif2020},
or conversely, work within a given approach, such as abstract interpretation~\cite{Albert_2005} or model checking~\cite{GurfinkelC03,MeolicFG04}.
The goal of our witnesses, however, is to expose  reasoning about concurrent programs
that can be transferred across the semantic divide.

Witness exchange aims to increase trust in the analysis results, and can effectively reveal
latent bugs in the analyzers.
Making program verification tools more dependable is an active area of research.
Novel testing methods have been developed specifically for testing analyzers~\cite{casso_testing_2021,klinger_differentially_2019,Kaindlstorfer2024,Fleischmann2024}.
Ideally, we would like to formally prove analyzers correct, and there has been
significant progress in this direction~\cite{jourdan_verasco_2015,cachera_certified_2010,franceschino_verified_2021,becker_ghost_2020,correnson_engineering_2023,de_vilhena_spy_2020,Stade2024};
however, none of these target the full range of advanced techniques needed for efficient analysis of real-world concurrent programs,
which are still under active research and development.

Ghost variables are crucial for the (relative) completeness of the proof systems of \citet{OwickiG76} and \citet{Lamport77}.
The usefulness of these and similar approaches was shown in case studies~\cite{Gries77,Dijkstra78b,Lubachevsky84,SemenyukD23}.
\citet{NipkowN99} formalize the Owicki-Gries method in Isabelle/HOL
and extend the method to parametric programs~\cite{Nieto01,Nieto02}.
Recent work extend the Owicki-Gries method to weak memory~\cite{LahavV15,DalvandiDDW19,DalvandiDDW22,WrightDBD23}
and persistent memory~\cite{RaadLV20,BilaDLRW22}.

We showed that our ghost instrumentation preserve the safety of the original program.
\citet{Filliatre2016} present an ML-style programming language and ensure non-interference via the type system,
allowing a bisimulation proof of ghost code erasure.
For concurrent code, \citet{Zhang2012} present a structural approach to establish erasure and \citet{Schmaltz2012} gives
bisimulation proofs for a significant portion of low-level C with ghost updates.

Our approach to generating ghost witnesses is reminiscent of a form of resource invariant synthesis~\cite{Clarke79}.
There are other approaches worth exploring, such as generating ghosts for counting proofs~\citet{FarzanKP14}
or inferring conditional history variables from counter-examples~\cite{VickM23}.
\citet{HoenickeMP17} explore how thread-modular reasoning can avoid ghosts by introducing proof systems
that consider interleavings of up to $k$ threads.
\citet{FarzanKP24} present an approach to simplifying proofs, reducing the complexity of the ghost state.

Thread-local concrete semantics have been used to justify thread-modular data-flow analysis~\cite{Mukherjee2017}
and concurrent separation logic~\cite{Brookes07,Brookes14,Krebbers0BJDB17,Ley-WildN13,NanevskiLSD14,SergeyNB15}.
For abstract interpretation, \citet{Mine2012} presents a thread-modular abstract interpretation framework as layers of abstractions of the interleaving semantics,
and in subsequent work~\cite{Mine2014,Mine2017a}, the interleaving semantics is encoded
into a local semantics by using auxiliary variables to track the control location of other threads.
In contrast, we have shown equivalence for a purely thread-local trace semantics where threads only
learn of other thread actions through synchronizing events.
\section{Conclusion}\label{sec:conclusion}
We have introduced a witness format for communicating information about concurrent programs between static analysis tools ---
bridging the divide between interleaving and thread-modular views of the concrete semantics.
We build on the notion of ghost variables, which is well-established for arguing about the correctness of concurrent programs.
The ghost witness can be easily encoded via a program transformation.
This way, the notion of a \emph{valid} witness is independent of the semantic view, which is crucial for an exchange format.
We extend an existing and established format to minimize the technical burden for other tools to support ghost witnesses.
For two types of thread-modular analyses implemented in the \Goblint\ abstract interpreter, we have shown how the invariants computed can naturally be expressed in the new format, highlighting that its expressiveness is useful for existing analyses.
A validation of the witnesses that  \Goblint\ generated for a set of concurrent SV-COMP tasks showed that most could be validated by the model checker \GemCutter\ building on an interleaving view of concurrency.
Altogether, the notion of ghost witnesses can be leveraged as a building block for exchanging more expressive invariants of concurrent programs between static analysis tools.
Future work may study how the results of other static analyzers for concurrent programs are best encoded via ghost code and how to cleverly exploit this information in thread-modular validators.
Additionally, it would be interesting to explore the applicability of ghosts in violation witnesses for concurrent programs.

\begin{credits}

\ifextendedversion
\subsubsection*{Data Availability Statement.}
A virtual machine that allows for reproduction of our experiments is available on Zenodo~\cite{artifact}.
The machine contains the source code and binaries of \Goblint and \UltimateGemCutter, the benchmark programs, and the detailed evaluation results underlying \cref{sec:validation-with-model-checking}.
\else
\fi
\subsubsection{\ackname}
This work was supported in part by the Deutsche Forschungsgemeinschaft (DFG) under project numbers 503812980 and 378803395/2428 \textsc{ConVeY}, the Shota Rustaveli National Science Foundation of Georgia (FR-21-7973),
and the European Union and the Estonian Research Council via project TEM-TA119.

\end{credits}

\bibliographystyle{splncs04nat}
\bibliography{bib/lit,bib/ghosts}

\end{document}